\documentclass[10pt]{article} 

\usepackage{jheppub} 
\usepackage[numbers,square,sort&compress]{natbib}

\usepackage[utf8]{inputenc}
\usepackage[english]{babel}
\usepackage{microtype} 
\usepackage{tikzsymbols}
\usetikzlibrary{external}
\usepackage{svg}
\svgpath{{figures/}}
\usepackage{amsmath,amssymb,amsfonts,amsthm,mathtools,bm}

\newtheorem{remark}{Remark}
\newtheorem{theorem}{Theorem}
\newtheorem*{theorem*}{Theorem}

\newtheorem{corollary}{Corollary}
\newtheorem{definition}{Definition}
\newtheorem{lemma}{Lemma}

\def\E{\mathbb{E}}

\newcommand{\SOMMA}[2]{\sum\limits_{#1}^{#2}}

\allowdisplaybreaks 

\usepackage[ruled,vlined]{algorithm2e} 
\SetKwProg{Fn}{Function}{:}{end}
\SetKw{KwBreak}{break}
\usepackage{algpseudocode}            

\usepackage{graphicx}  
\graphicspath{{Images/}}
\usepackage{xcolor}    
\usepackage{float}     
\usepackage{booktabs}
\usepackage{longtable}
\usepackage{enumitem}
\usepackage{tikz}
\usepackage{pgfplots}
\usepackage{svg}
\usepackage{subcaption}

\usetikzlibrary{
    positioning,
    fit,
    calc,
    decorations.pathreplacing,
    arrows.meta,
    decorations.text,
    backgrounds,
    shapes.symbols,
    decorations.pathmorphing,
    shapes.geometric,
    shadows,
    }

\usepackage{cleveref} 
\usepackage{csquotes} 

\let\oldbibliography\thebibliography
\renewcommand{\thebibliography}[1]{%
  \oldbibliography{#1}%
  \setlength{\itemsep}{0pt}
  \setlength{\parskip}{0pt}
  \linespread{0.9}\selectfont
  \scriptsize
}

\usepackage{setspace}
\title{Classifications in modular restricted Boltzmann machines}

 \author[a,b,c]{Elena Agliari}%
 \author[a,b]{, Andrea Lepre}%
 \author[a]{, Edoardo Roscani}
 \affiliation[a]{Dipartimento di Matematica, Sapienza Università di Roma, Italy}
 \affiliation[b]{GNFM, Istituto Nazionale di Alta Matematica Francesco Severi (INdAM), Roma, Italy}
 \affiliation[c]{CNR-Nanotec, Unità di Lecce, Lecce, Italy}

\abstract{We consider a modular associative neural network made of $L$ Hopfield models (HMs), coupled so that intra-module interactions are Hebbian and inter-module interactions are anti-Hebbian; this competitive coupling is known to endow the network with pattern-disentanglement capabilities. The integral representation of this system coincides with an assembly of $L$ restricted Boltzmann machines (RBMs) whose hidden layers are coupled, thereby extending the HM-RBM duality to the modular setting. We then train this modular RBM, via one-step contrastive divergence, to perform a classification task in which a query encoded on the visible layers is mapped onto an $L$-tuple of labels read off the hidden layers. When the query is composed of $L$ patterns that are mutually orthogonal on average, we prove that the RBM weights obtained as empirical means over the training dataset, as suggested by the HM-RBM equivalence, constitute a fixed point of the learning dynamics, and we derive an explicit, non-asymptotic bound on the residual drift. We then turn to a harder classification task in which each module is queried with a mixture of $L$ patterns and we show numerically that the same setting for the RBM weights still provides an effective set-up, letting the network jointly classify and disentangle the mixture.}

\begin{document}

\maketitle

\section{Introduction}

Among the various approaches to the theory of learning, those grounded in statistical mechanics occupy a special place: by casting a neural network as a disordered spin system, its ability to retrieve, generalise, or classify can be studied with the same toolbox developed for spin glasses, often yielding closed-form, testable predictions rather than purely empirical ones, see e.g., \cite{Lenka-2017}. The paradigm of this programme is the Hopfield model (HM) of associative memory~\cite{hopfield1982}, in which a set of patterns is stored through a Hebbian prescription and retrieved as the minima of a quenched-disorder energy landscape~\cite{amitgutfreundsompolinsky1985}. A closely related, and equally paradigmatic, architecture is the restricted Boltzmann machine (RBM)~\cite{ackleyhintonsejnowski1985}, a bipartite (in the simplest version) energy-based model in which a hidden representation is learnt from data by minimising a Kullback-Leibler divergence via, e.g., the contrastive divergence~\cite{hinton2002cd,10.1093/oso/9780198530237.001.0001}. Remarkably, the two systems are, under mild conditions, one and the same object~\cite{BARRA20121,Agliari_2012PRL, Mezard-PRE2017, TubianaMonasson-PRL2017, DecelleFurtlehner-2021, Roudi-EPL2025}. The connection can be made stronger by equipping the Hopfield side with a ``supervised Hebbian rule'' \cite{Alemanno_2023}: assuming that the network has access only to a noisy sample of the patterns to be retrieved and that their label is known (whence the adjective supervised), class-wise averages can replace the ground-truth patterns in Hebb's rule, then, the resulting weights in the equivalent RBM can be shown to constitute an asymptotically stationary point of the contrastive-divergence dynamics \cite{LEONELLI2021314}. In general cases, this set-up still constitutes an efficient initialization for training \cite{SmartZilman-ICLR2021}.

Real data, however, are rarely single, isolated objects: a musical chord is a superposition of notes, an image a composition of several visual concepts, a document a mixture of topics. Reproducing this compositional structure at the level of an associative network requires several HMs that cooperate, or compete, so as to specialise on the individual constituents of the signal. Recently, a modular architecture in which $L$ HMs interact through intra-module Hebbian couplings and inter-module anti-Hebbian, competitive couplings was introduced and shown to spontaneously accomplish pattern disentanglement \cite{AGLIARI2026108181,agliari2025multi,agliari2025generalized}: presented with a mixture of $L$ patterns, replicated identically across the $L$ modules, the network relaxes to a configuration in which each module retrieves a single, distinct constituent of the mixture. The model was subsequently extended to the case in which the patterns to be disentangled are not directly accessible but only known through a finite sample of noisy examples; there, the disentanglement region was shown to gradually shrink as the quality of the dataset degrades \cite{AGLIARI2026131134}. 

In this work, we turn to the equivalent representation of these modular networks, namely, an assembly of RBMs coupled through their hidden neurons, and investigate its learning dynamics. In particular, we show that the weight configuration suggested by the correspondence with the modular associative model goes beyond the formal equivalence between the two descriptions, being an interpretable, attracting point of a genuine learning dynamics. 

To reach this goal we proceed as follows. In Section~\ref{sec:Model}, we review the HM-RBM equivalence in the modular setting and we specialize on a classification task, in which a query encoded by the visible layers is to be mapped onto an $L$-tuple of labels encoded by the hidden layers; next we specify two datasets of increasing difficulty on which the machine can be trained. In the former, each module is presented with a noisy example of a distinct ground-truth; in the latter, every module is presented with the same noisy mixture of all $L$ ground-truths, so that classification must be accompanied by disentanglement.
Then, in Section~\ref{sec:Learning_Phase} we derive the contrastive-divergence learning rule for the intra-module weights of this modular RBM and, restricting to the simpler of the two datasets, we prove a fixed-point theorem: the class-wise empirical-mean weights suggested by the associative-network side of the duality are, on average, stationary under one step of contrastive divergence, with a bounded residual drift. This result generalises and makes more accurate the asymptotic single-module statement of~\cite{LEONELLI2021314} to any number of coupled modules. Finally, in Section~\ref{sec:Numerical} we verify this bound numerically and we check that the planted weights do solve the classification task they were designed for. We then go beyond what the theorem guarantees: running the learning rule from
weights that carry no information whatsoever about the archetypes, either identically zero or Gaussian with increasing variance, we find that it converges to the empirical-mean configuration in every case, so that the latter is not merely
stationary but an attractor of contrastive divergence, with a basin of attraction wide enough to accommodate initialisations whose norm is comparable to the signal itself. Also, we move to the harder dataset and we show numerically that, although an exact analytical result for the empirical-mean weights is unavailable, they remain an effective set-up, letting the modular RBM jointly classify and disentangle the mixture across the region of parameter space predicted by~\cite{AGLIARI2026131134}. Proofs and auxiliary technical results are collected in Appendices~\ref{App:A} and~\ref{App:B}.

\section{The Model} \label{sec:Model}
Let us consider $L$ modules of $N$ binary neurons each, with neuron state denoted by
$\bm \sigma^a = (\sigma^a_1, \dots, \sigma^a_N) \in \{-1, + 1\}^N$ with $a=1, ...,L$, and $K$ binary patterns denoted by $\bm \xi^\mu = (\xi_1^\mu, \dots, \xi^\mu_N) \in \{-1,+ 1\}^N$, with $\mu = 1, \dots, K$, whose entries are sampled from the probability distribution
\begin{equation}\label{eq:pattern_prob}
    \mathcal{P}(\xi_i^\mu) = \dfrac{1}{2}\delta(\xi_i^\mu - 1) + \dfrac{1}{2}\delta(\xi_i^\mu + 1). 
\end{equation} 
Then, any global configuration $\boldsymbol \sigma = (\bm \sigma^1, ..., \bm \sigma^L) \in \{-1, +1 \}^{N\times L}$ is associated to a ``cost'' or ``energy'' given by 
\begin{equation} \label{eq:LAM_hamiltonian}
    \mathcal{H}^{\mathrm{HM}}(\bm\sigma;\bm g, \bm\xi):= - \dfrac{1}{2N}\SOMMA{\mu=1}K\SOMMA{a,b=1}L g_{ab}\SOMMA{i,j=1}N \sigma_i^a\xi_i^\mu\xi_j^\mu\sigma_j^b,
\end{equation}
where both inter- and intra-module couplings contribute and
where $\bm g$ is a $L \times L$ matrix that tunes their relevant magnitudes. Specifically, here we focus on the structure
\begin{align} \label{eq:q}
        g_{ab} := \begin{cases}
             \ \ \ 1 \  \text{if} \ \ a=b \\
            -\lambda \ \text{if} \ \ a \not = b
        \end{cases}
    \end{align}
        with $\lambda \in [0, (L-1)^{-1})$ ensuring that $\bm g$ is positive defined; see Fig. \ref{fig:LAM_Model} for a schematic representation of the network.
    \begin{figure}[t!]
      \centering
      \includegraphics[width=0.6\textwidth]{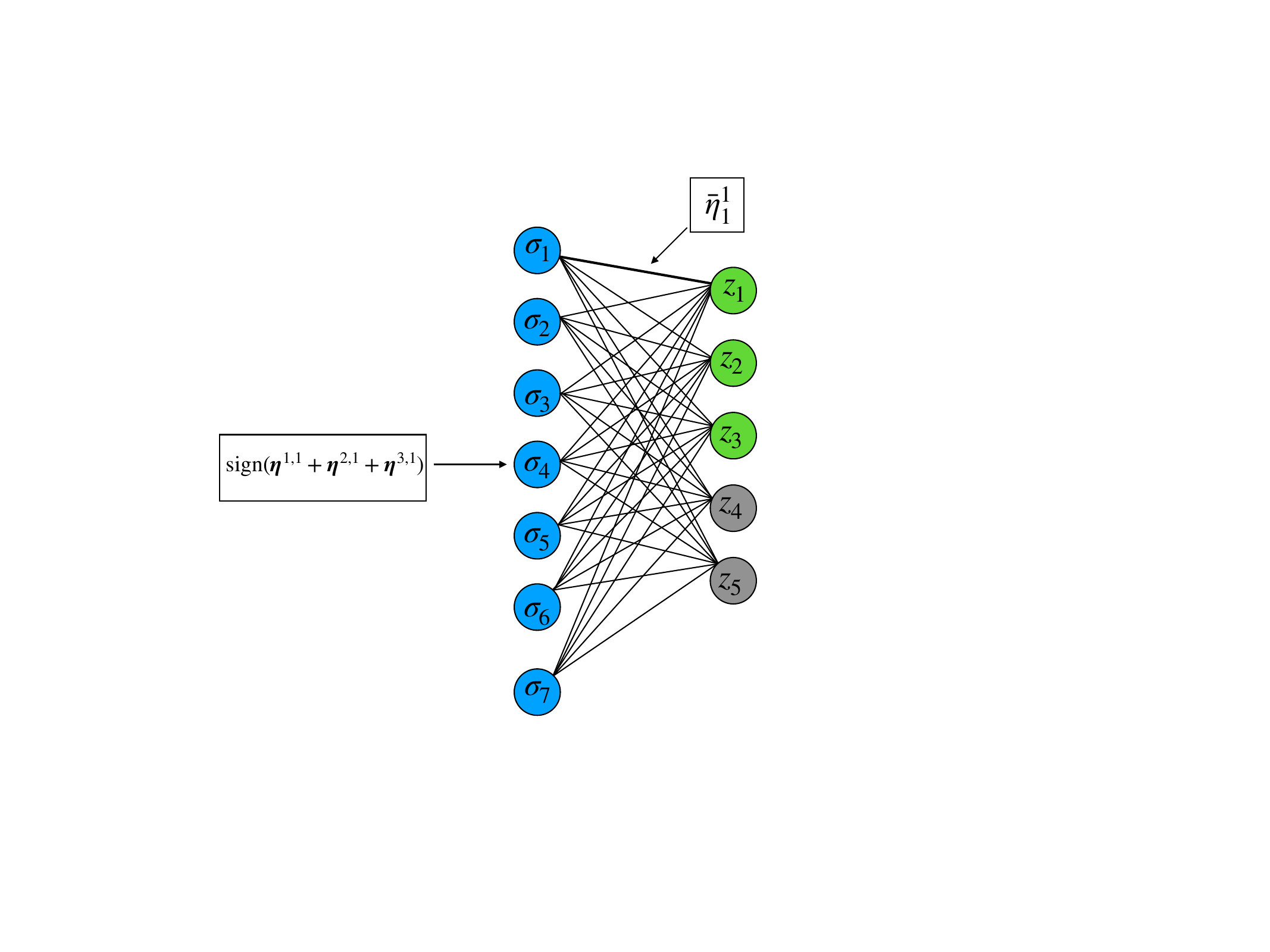}
      \caption{Example of modular associative model with $L = 3$ modules and $N = 6$ neurons per module. 
In each module $a=1, ...,L$, the neurons 
$\bm \sigma^a = (\sigma^a_1, \ldots, \sigma^a_N) \in \{-1,+1\}^N$ are fully-connected, with intra-module Hebbian interactions built on a set of $K$ patterns $\{\bm \xi^\mu\}_{\mu=1,...,K}$ (or on a sample of corrupted copies thereof). 
The $L$ modules are coupled to one another through inter-module inhibitory connections 
of strength $-\lambda$, which penalise configurations in which different modules 
retrieve the same pattern, thereby promoting the simultaneous retrieval of 
$L$ distinct patterns across the $L$ modules.}
      \label{fig:LAM_Model}
    \end{figure}
%
%
    %
    Moreover, we adopt a Boltzmann-Gibbs distribution that assigns to each configuration $\bm \sigma$ a probability
    \begin{equation}
         \mathcal{P}^{\mathrm{HM}}(\bm\sigma ; \beta, \bm g, \bm\xi) := \dfrac{e^{-\beta\mathcal{H}^{\mathrm{HM}}(\bm\sigma;\bm\xi,\bm g)}}{\mathcal{Z}^{\mathrm{HM}}(\beta, \bm g, \bm\xi)}
    \end{equation}
    with the partition function $\mathcal{Z}^{\mathrm{HM}}$  defined as
    \begin{align}
    \mathcal{Z}^{\mathrm{HM}}(\beta, \bm g, \bm\xi) := \SOMMA{\{\bm\sigma^1\}}{2^N} \cdots \SOMMA{\{\bm\sigma^L\}}{2^N} \exp\biggl\{\dfrac{\beta}{2N}\SOMMA{a,b=1}{L}g_{ab}\SOMMA{\mu,i,j=1}{K,N,N}\xi_i^\mu\xi_j^\mu\sigma_i^a\sigma_j^b\biggr\},
\end{align}
and $\beta := 1/T$ tuning the probability width (for $\beta \to 0$ the system will assign random activations to neurons, while for $\beta \to \infty$ the 
system drives neurons deterministically towards the energy minima). 

Remarkably, as a consequence of eqs.~\eqref{eq:LAM_hamiltonian}-\eqref{eq:q}, configurations where all modules are aligned with the same pattern (e.g., $\boldsymbol \sigma = (\boldsymbol \xi^1, ..., \boldsymbol \xi^1)$) are relatively unlikely, whereas configurations in which different modules align with distinct patterns (e.g., $\boldsymbol \sigma = (\boldsymbol \xi^1, ..., \boldsymbol \xi^L)$) are comparatively favored. This property underlies the ability of the system to perform disentanglement tasks \cite{AGLIARI2026108181}: when initialized with a mixture state (e.g., for each module $\boldsymbol \sigma^a = \textrm{sign}(\sum_{\mu=1}^L \boldsymbol \xi^{\mu})$), under suitable conditions over $\lambda$ and $\beta$, Gibbs dynamics drives the system to spontaneously relax towards a state in which each individual module retrieves a distinct component of the composite signal (e.g., $\boldsymbol \sigma^a =  \boldsymbol \xi^{a}$ for $a=1, ...,L$). 

Moreover, the model has been recently generalized in \cite{AGLIARI2026131134} to face scenarios where the patterns $\boldsymbol \xi$ to be reconstructed are not directly accessible, instead, $M$ corrupted copies denoted as $\boldsymbol \eta$ are available. In this context, the latter are referred to as \textit{examples} and meant as ``experimental measurements'' of the original patterns $\boldsymbol \xi$, referred to as ground-truths or \textit{archetypes}. To retain a mathematical control, examples are sampled as
\begin{equation} \label{eq:esempi}
    \eta_i^{\mu, m} := \xi_i^\mu \, \chi_i^{\mu, m}, \quad i = 1, \ldots, N, \quad \mu = 1, \ldots, K, \quad m = 1, \ldots, M,
\end{equation}
where $\chi_i^{\mu,m}$ is a random variable drawn from the distribution:
\begin{equation}
    \mathcal{P}(\chi_i^{\mu,m})= \frac{1+r}{2} \delta(\chi_i^{\mu,m} - 1) + \frac{1-r}{2}\delta(\chi_i^{\mu,m} + 1), \quad \forall \ i,\mu,m, \\
\end{equation}
with $r \in (0,1]$ acting as the dataset quality parameter.
The informativeness of the accessible set $\boldsymbol \eta$ about the target set $\boldsymbol \xi$ is quantified by the parameter 
\begin{equation}
\label{eq:rho_def}
    \rho:=\dfrac{1- r^2}{Mr^2},
\end{equation}
also referred to as dataset \emph{entropy} \cite{Alemanno_2023}. Assuming that for each example the label $\mu$ is revealed (i.e., in a supervised setting), we can replace the unknown archetypes with their estimators 
\begin{equation}\label{eq:etabar_def}
    \bar {\boldsymbol \eta}^\mu := \dfrac{1}{r\,M}\SOMMA{m=1}{M} {\boldsymbol \eta}^{\mu,m},  \quad \mu = 1, \ldots, K,
\end{equation}
whose entries satisfy\footnote{The division by $r$ is what makes $\bar{\bm\eta}^\mu$ an \emph{unbiased} estimator of the archetype it stands for. Had we kept the plain sample mean, its expectation would have been $r\,\bm\xi^\mu$, and the systematic shrinkage by $r$ would have propagated into every quantity built on the weights. It is worth comparing this choice with the one adopted in \cite{AGLIARI2026131134}, where the couplings are built on $\tilde{\bm\eta}^\mu/\sqrt{\mathcal R}$, with $\tilde{\bm\eta}^\mu:=\frac1M\sum_m \bm\eta^{\mu,m}$ the biased estimator and with $\mathcal R:=r^2+(1-r^2)/M$ the related second moment. Each normalization is well suited to the framework in which it is employed. Dividing by $\sqrt{\mathcal R}$ normalizes the \emph{coupling}: since $\frac1N\|\tilde{\bm\eta}^\mu\|^2\to\mathcal R$, the resulting weights have unit norm, so that the scale of the energy, and with it the meaning of $\beta$, is the same at every $\rho$, which is a natural convention when constructing a phase diagram. Dividing by $r$ normalizes instead the \emph{estimator}, which is the natural convention here, where the weights are the object being learnt and are compared, through the learning rule, with a target living on a prescribed scale.}
\begin{equation}\label{eq:etabar_moments}
    \E_{\bm \chi}\bigl[\bar\eta^\mu_i\bigr]=\xi^\mu_i,
    \qquad
    \E_{\bm \chi}\bigl[\bigl(\bar\eta^\mu_i-\xi^\mu_i\bigr)^2\bigr]=\rho,
\end{equation}
where the average being over the corruption variables $\bm \chi$ of \eqref{eq:esempi} at fixed archetypes, and the dataset entropy is exactly the mean-square error of the estimator, thus, $\bar{\bm\eta}^\mu$ is an accurate proxy of $\bm\xi^\mu$ as soon as $\rho\ll1$, that is $Mr^2\gg1$. We can then revise the Hamiltonian \eqref{eq:LAM_hamiltonian} as 
\begin{equation}\label{eq:Hamiltonian_Sup}
    \mathcal{H}^{\mathrm{HM}}(\bm\sigma;\bm g, \bar{\bm\eta}):= - \dfrac{1}{2N}\SOMMA{\mu=1}K\SOMMA{a,b=1}L g_{ab}\SOMMA{i,j=1}N \sigma_i^a\bar \eta_i^\mu\bar \eta_j^\mu\sigma_j^b.
\end{equation}
As shown in \cite{AGLIARI2026131134}, there still exists a region in the parameter space $(\lambda, \beta)$ where this system is able to disentangle a composite signal and reconstruct the constituting archetypes; however, this region progressively shrinks as $\rho$ gets larger.

\begin{remark}
    The quality parameter $r$, which enters the definition \eqref{eq:etabar_def} is a property of the dataset and it is estimable from the dataset alone, with no access to the archetypes. Indeed, denoting by $\tilde\eta^\mu_i:=\frac1M\sum_m\eta^{\mu,m}_i$ the biased estimator, one has $\E[(\tilde\eta^\mu_i)^2]=r^2+(1-r^2)/M$, so that
\begin{equation}\label{eq:r_hat}
    \hat r^{\,2}:=\frac{M\,\widehat{\mathcal R}-1}{M-1},
    \qquad
    \widehat{\mathcal R}:=\frac{1}{KN}\SOMMA{\mu=1}{K}\SOMMA{i=1}{N}\bigl(\tilde\eta^\mu_i\bigr)^2 ,
\end{equation}
is an unbiased estimator of $r^2$ built exclusively on the observed examples, and, being an average of $KN$ bounded and weakly dependent terms, it concentrates around its mean at the rate $(KN)^{-1/2}$. 
\end{remark}

\begin{figure}[t!]
    \centering
    \includegraphics[width=0.65\textwidth]{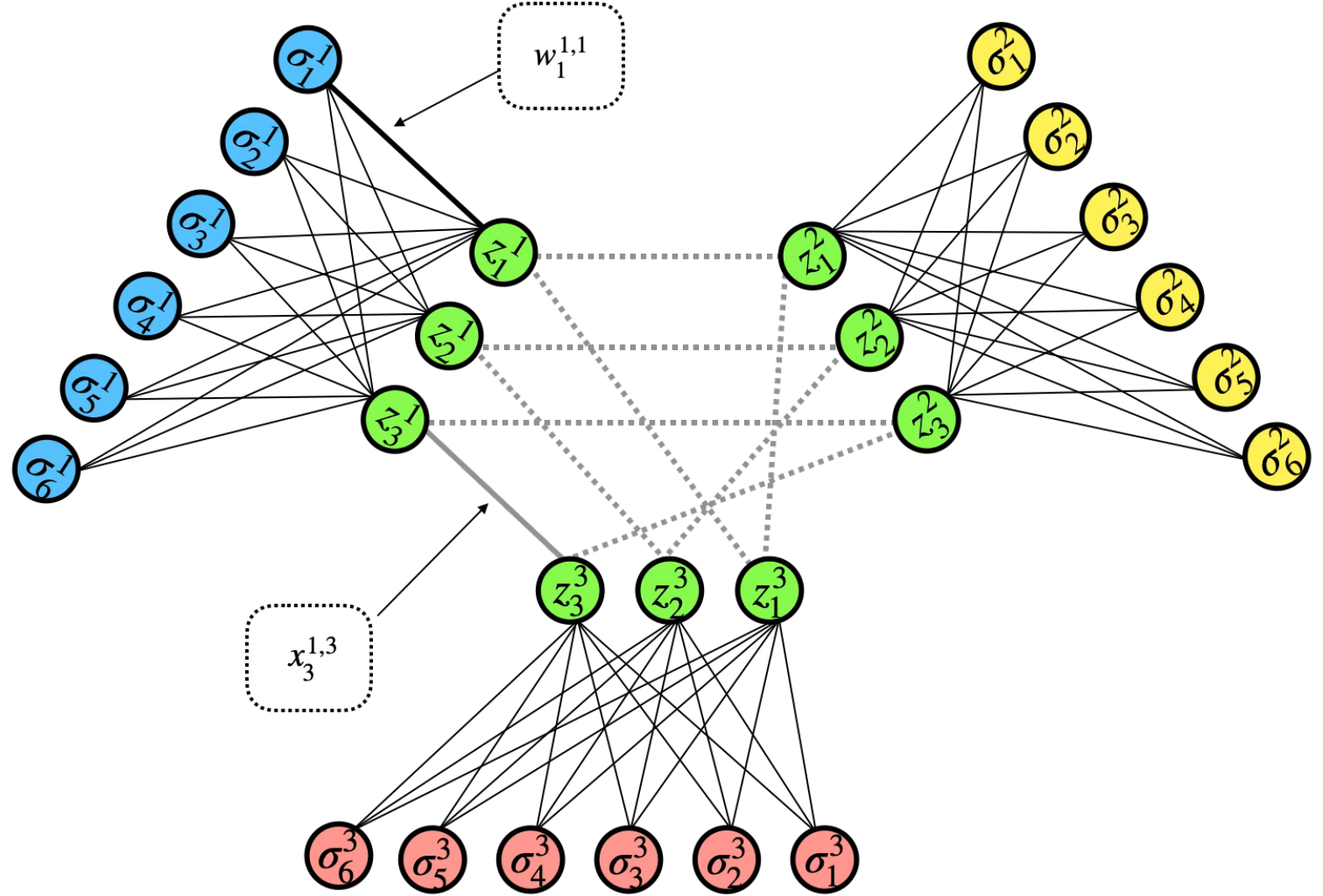}
    \caption{Architecture of the modular RBM equivalent to the modular HM in Fig.~\ref{fig:LAM_Model}. It exhibits $L=3$ modules,
    $K=3$ hidden neurons and $N=6$ visible neurons per module.
    In fact, each module $a \in \{1,2,3\}$ consists of a visible layer
    $\boldsymbol{\sigma}^{a} = (\sigma^{a}_{1},\ldots,\sigma^{a}_{N}) \in \{-1,+1\}^{N}$
     and a hidden layer
    $z^{a} = (z^{a}_{1},\ldots,z^{a}_{K}) \in \mathbb{R}^{K}$.
    Black edges represent the intra-module weights $w^{\mu,a}_{i}$, coupling the
    $i$-th visible neuron of module $a$ to the $\mu$-th hidden neuron of the same module.
    Dashed edges denote the inter-module couplings $x^{ab}_{\mu}$, connecting
    hidden neurons with the same class index $\mu$ across different modules.}
    \label{fig:Modular_RBM}
\end{figure}

\subsection{Equivalence with restricted Boltzmann machines}

As shown in \cite{AGLIARI2026108181}, the modular Hebbian network described by \eqref{eq:LAM_hamiltonian} exhibits an equivalent representation in terms of a combination of $L$ interacting RBMs, see Fig. \ref{fig:Modular_RBM}. More specifically, each RBM consists of a Bernoulli--Gaussian bipartite network with a layer (also referred to as ``visible'') containing $N$ binary neurons denoted as $\boldsymbol \sigma^a \in \{-1, +1 \}^N$, and with a layer (also referred to as ``hidden'') containing $K$ real-valued neurons $\boldsymbol z^a \in \mathbb{R}^K$, endowed with a Gaussian prior, for $a=1, ...,L$. Interactions occur both between neurons in different layers, through the weight matrix $\boldsymbol w^{a} \in \mathbb R^{N\times L}$ ($a = 1, \dots, L$), and between real-valued neurons belonging to different modules, through the coupling matrix $\boldsymbol x^{a} \in \mathbb R ^{K \times L}$ ($a = 1, \dots, L$), as described by the following cost function:
\begin{equation} \label{eq:Hamiltonian_RBMM}
    \mathcal{H}^{\text{RBM}}(\bm \sigma, \bm z| \bm w) := -\dfrac{1}{\sqrt{N}}\SOMMA{a=1}{L}\SOMMA{\mu=1}{K}\SOMMA{i=1}{N}w_i^{\mu,a}\sigma_i^az_\mu^a + \dfrac{1}{2}\SOMMA{\mu=1}{K}\SOMMA{a,b=1}{L,L}z_\mu^a x^{ab}_{\mu}z_\mu^b.
\end{equation}

The equivalence is proven by showing that, when the RBM weights coincide with the stored patterns, the partition functions of the two models coincide. In fact, setting $w_i^{\mu,a} = \xi_i^{\mu}$ and $x_{\mu}^{ab} = (g^{-1})_{ab}$ for any $i, \mu, a, b$ and applying a Hubbard-Stratonovich transformation
\begin{align}
\mathcal{Z}^{\text{HM}}(\beta, \boldsymbol{g}, \boldsymbol{\xi}) 
&= \sum_{\{\bm \sigma^1\}} \cdots \sum_{\{\bm \sigma^L\}} 
\exp\left(\frac{\beta}{2N} \sum_{a,b=1}^{L,L} g_{ab} \sum_{\mu,i,j=1}^{K,N,N} 
\xi_i^{\mu} \xi_j^{\mu} \sigma_i^{a} \sigma_j^{b}\right) \notag \\
&= \sum_{\bm\sigma^1... \bm\sigma^L}\int_{\mathbb{R}^{LK}} \prod_{\mu=1}^{K} \dfrac{\beta^{L/2}}{(2\pi)^{L/2}\sqrt{\det \bm g}}\, d^L\bm z_\mu\;\exp\left(-\dfrac{\beta}{2}\SOMMA{\mu,a,b}{K,L,L}z_\mu^a (g)^{-1}_{ab}z_\mu^b + \dfrac{\beta}{\sqrt{N}}\SOMMA{\mu,a,i}{K,L,N}\xi_i^\mu \sigma_i^a z_\mu^a\right)  \notag \\
&= \mathcal{Z}^{\text{RBM}}(\beta, \boldsymbol{g}, \boldsymbol{\xi}). \notag
\end{align}
This result constitutes a straightforward generalization of the equivalence between the HM and a single RBM (see, e.g., \cite{BARRA20121,Agliari_2012PRL}), which is recovered as $L=1$ or as $\lambda =0$ (in that case we recover $L$ independent HMs/RBMs). Remarkably, the equivalence between the two partition functions implies that the thermodynamics of the two models coincide, since 
the partition function plays as the generating function of thermodynamic observables (e.g., free energy, internal energy, entropy, magnetization).
\begin{figure}[t!]
    \centering
    \includegraphics[width=0.8\linewidth]{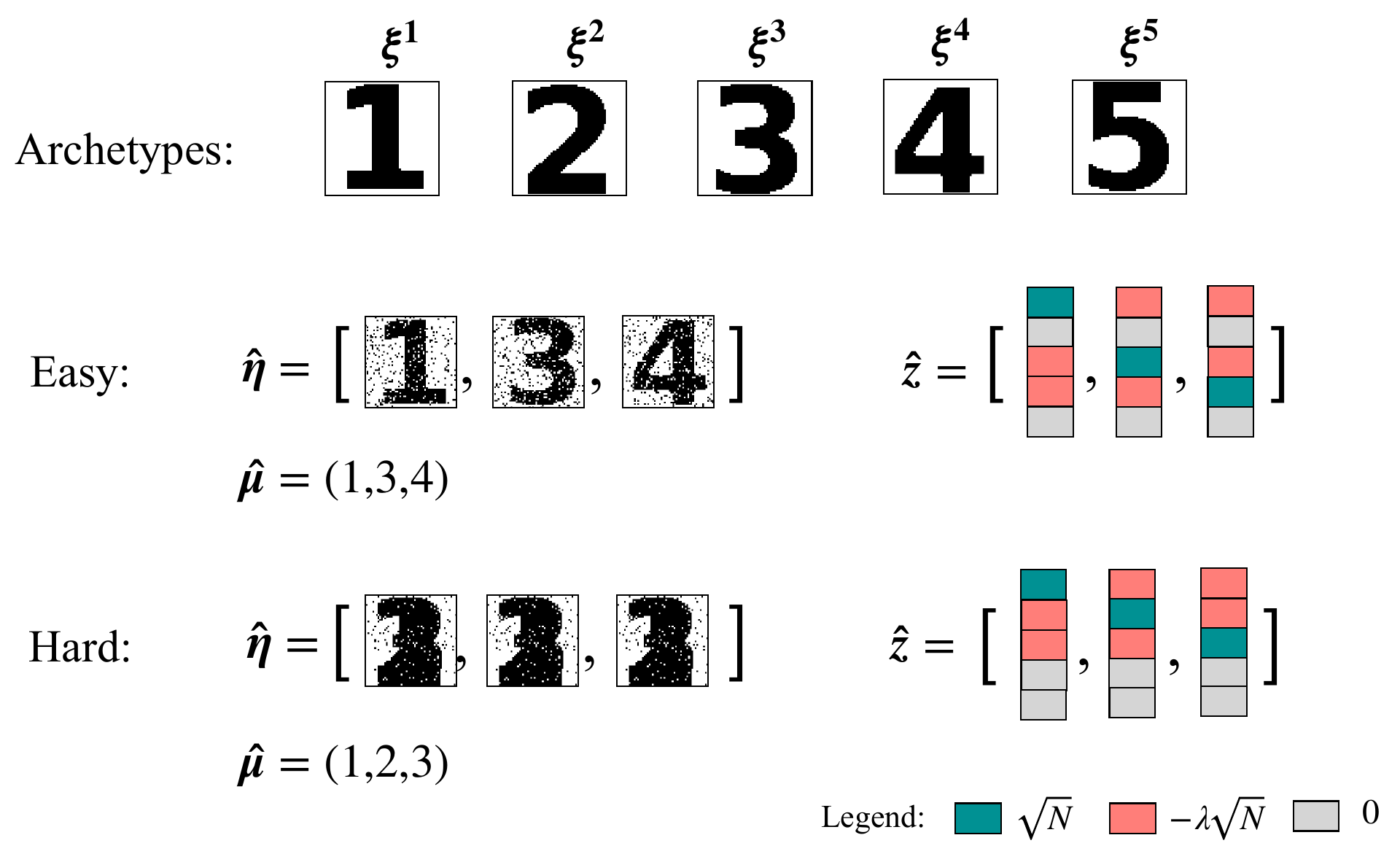}
    \caption{Schematic representation of the two supervised datasets introduced in Sec.~\ref{sec:Model}, drawn for $K=5$ archetypes and $L=3$ modules. The top row shows the archetypes $\bm\xi^1,\dots,\bm\xi^5$, here, for illustrative purposes, taken as five digits. A training item is a pair $(\hat{\bm\eta},\hat{\bm z})$, where $\hat{\bm\eta}$ collects the $L$ vectors presented to the $L$ visible layers, while $\hat{\bm z}$ collects the $L$ label codes \eqref{eq:clamp}. In Case 2a, simple classification, the $L$ visible layers receive examples of $L$ distinct archetypes, here $\hat{\bm\mu}=(1,3,4)$, and each module has to classify the archetype it is shown. In Case 2b, hard classification, the $L$ visible layers all receive the same vector, the mixture of $L$ distinct examples, here $\hat{\bm\mu}=(1,2,3)$, so that a single mixture is presented $L$ times and the modules have to disentangle it, each classifying a different constituent. The label codes are built in the same way in the two cases, as recalled in the legend.}
    \label{fig:schematic_dataset}
\end{figure}
However, this setup provides no information about the machine's training process. 
To fill this gap we need to specify the kind of task we would like to train the machine for, and accordingly revise the framework. Here, we focus on a classification task such that, given an item (query) represented by a binary vector $\boldsymbol{\hat \eta} \in \{-1, +1\}^{N\times L}$, this is associated to a $L$-tuple label $\hat {\bm\mu} = (\mu_1, ..., \mu_L) \in \{1, ..., K\}^L$. The former will be straightforwardly encoded by the visible layers, and the latter by the hidden layers. In particular, we choose the following mapping\footnote{The factor $\sqrt N$ and the negative entries in
\eqref{eq:clamp} correspond to the value the hidden layer takes when
the machine classifies correctly, as follows from the conditional law discussed in Section~\ref{sec:Learning_Phase}.} 
$\hat {\bm\mu} \to \boldsymbol{\hat z} \in \mathbb R^{K \times L}$: 
\begin{equation}
\label{eq:clamp}
    \hat z^a_\mu = \sqrt{N}\biggl(\delta_{\mu,\mu_a} - \lambda \sum_{c \not =a}\delta_{\mu,\mu_c}\biggr) 
    , \quad \forall \mu, a.
\end{equation}
In words, when clamped $\hat{\boldsymbol z}^a$, the $a$-th layer highlights the single label $\mu_a$ by the unique positive entry $\hat z^a_{\mu_a}$, the remaining hidden neurons of the layer taking a negative or null value.

Thus, the training dataset is given by 
\begin{equation}\label{eq:dataset_S}
    \mathcal{S} := \bigl\{(\hat{\bm\eta}^{\,p},\,\hat{\bm z}^{\,p})\bigr\}_{p=1,\dots,P},
\end{equation} 
with $P$ the number of items in the dataset. In the following, we will specialize $\hat {\bm z}^p$ and $\hat{\bm \eta}^{p}$ as follows (see also Fig.~\ref{fig:schematic_dataset}): 
\begin{itemize}
\item[1.] 
For each $p=1, ...,P$, $\hat {\bm \mu}^p =(\hat \mu_1^p, ...,\hat \mu_L^p)$ is a tuple of $L$ distinct labels, and $\hat {\bm z}^p$ follows from \eqref{eq:clamp};
\item[2a.] [Simple classification]
$\hat{\bm \eta}^{p} = (\bm \eta^{\hat \mu_1^p}, ..., \bm \eta^{\hat \mu_L^p})$, namely each query is a sequence of $L$ examples corresponding to $L$ distinct archetypes $\bm \xi^{\hat \mu_1^p}, ..., \bm \xi^{\hat \mu_L^p}$, and defined according to \eqref{eq:esempi};
\item[2b.] [Hard classification] $\hat{\bm \eta}^{p} = (\bm \eta^{\hat {\bm\mu}^p}, ..., \bm \eta^{\hat {\bm\mu}^p})$, where $\bm \eta^{\hat {\bm\mu}^p}= \textrm{sign}(\sum_{\ell=1}^L \bm \eta^{\hat {\mu}^p_{\ell}})$, namely each query is a sequence of repeated spurious configurations combining $L$ examples, defined according to \eqref{eq:esempi}.
\end{itemize}

The case 2a will be addressed analytically and, in particular, in the next section we will prove that, under suitable conditions and by posing $\bm x = \bm g^{-1}$, the learning rules for $\bm w$ meant to train the machine to accomplish the classification task, display as fixed points $\bm w = \bar {\bm \eta}$. As a result, 
\begin{align}
\mathcal{Z}^{\text{RBM}}(\beta, \boldsymbol{g}, \boldsymbol{\bar \eta})
&= \sum_{\bm\sigma^1... \bm\sigma^L}\int_{\mathbb{R}^{LK}} \prod_{\mu=1}^{K} \dfrac{\beta^{L/2}}{(2\pi)^{L/2}\sqrt{\det \bm g}}\, d^L\bm z_\mu\;\exp\biggl\{-\dfrac{\beta}{2}\SOMMA{\mu,a,b}{K,L,L}z_\mu^a (g)^{-1}_{ab}z_\mu^b \notag + \dfrac{\beta}{\sqrt{N}}\SOMMA{\mu,a,i}{K,L,N}\bar \eta_i^\mu \sigma_i^a z_\mu^a\biggr\} \\
&= \sum_{\{\sigma^1\}} \cdots \sum_{\{\sigma^L\}} 
\exp\bigl\{\frac{\beta}{2N} \sum_{a,b=1}^{L,L} g_{ab} \sum_{\mu,i,j=1}^{K,N,N} 
\bar \eta_i^{\mu} \bar \eta_j^{\mu} \sigma_i^{a} \sigma_j^{b}\bigr\} \notag \\
&= \mathcal{Z}^{\text{HM}}(\beta, \boldsymbol{g}, \boldsymbol{\bar \eta}),
\end{align}
that is precisely the partition function of the supervised modular networks \eqref{eq:Hamiltonian_Sup}. 

This equivalence suggests that the setting
\begin{equation} \label{eq:setting}
    w_i^{\mu,a} = \bar \eta_{i}^\mu, \quad i=1, \dots, N, \quad \mu=1, \dots, K \quad a=1, \dots, L
\end{equation}
could also be a meaningful one for Case 2b, as corroborated numerically in Sec.~\ref{sec:Dinsentanglement}.

\section{Analytical Results} \label{sec:Learning_Phase}
Let us consider the cost function of the modular RBM~\eqref{eq:Hamiltonian_RBMM}
and set the inter-module couplings to $\bm x=\bm g^{-1}$, while the intra-module
weights $\bm w$ are treated as the only learnable parameters \footnote{This choice follows from the equivalence established in Section~\ref{sec:Model}: the structure of the matrix $\bm g$ defines the architectural fingerprint by penalising shared archetype retrieval across modules. 
Furthermore, the homogeneous off-diagonal entries mirror the model symmetry, which assigns identical relevance to all modules, while the magnitude $\lambda$ can ultimately be absorbed into $\bm w$.}.
The question we address in this section is whether a
standard contrastive-divergence rule, acting on $\bm w$ alone, drives the machine
to the weight configuration $\bm w=\bar{\bm\eta}$ that makes it coincide with the
supervised modular network of eq.~\eqref{eq:Hamiltonian_Sup}.

\subsection{Learning rules} \label{sec:learning_rules}
Let us recall, from Sec.~\ref{sec:Model}, the supervised training set 
$\mathcal{S} = \bigl\{(\hat{\bm\eta}^{\,p},\,\hat{\bm z}^{\,p})\bigr\}_{p=1,\dots,P}$,
where each item $\hat{\bm\eta}^{\,p}=(\bm\eta^{\hat\mu_1^p},\dots,\bm\eta^{\hat\mu_L^p})$
is a sequence of $L$ examples of $L$ distinct archetypes (Case~2a), and the label
$\hat{\bm z}^{\,p}$ is given by Eq.~\eqref{eq:clamp}. 
The goal of learning is to capture the distribution underlying the data in $\mathcal{S}$. Assuming the items of $\mathcal{S}$ to be i.i.d., we estimate this distribution through the empirical measure, i.e., the uniform measure over the dataset,
\begin{equation}
    \label{eq:Q_def}
    \mathcal Q(\bm \sigma, \bm z) := \dfrac1P\SOMMA{p=1}{P}
    \delta\bigl(\bm\sigma-\hat{\bm\eta}^{\,p}\bigr)\,
    \delta\bigl(\bm z-\hat{\bm z}^{\,p}\bigr),
\end{equation}
the two deltas being understood on $\{-1,+1\}^{N\times L}$ and on $\mathbb R^{K\times L}$ respectively, and \eqref{eq:Q_def} reducing at $L=1$ to the target law adopted for a single machine in \cite{Alemanno_2023}.

A natural loss function to be minimised during the update of
the model's parameters is the Kullback--Leibler divergence
\begin{equation}\label{eq:KL_def}
    D_{KL}(\mathcal Q \,\|\, \mathcal P) := \sum_{\{\boldsymbol{\sigma}\}} \int d\bm z\;
    \mathcal Q(\boldsymbol{\sigma}, \bm z)\,
    \ln \frac{\mathcal Q(\boldsymbol{\sigma}, \bm z)}{\mathcal P(\boldsymbol{\sigma}, \bm z)} .
\end{equation}
The weights are then updated by gradient descent and we get \cite{10.1093/oso/9780198530237.001.0001}
\begin{equation}\label{eq:KL_grad}
    \frac{\partial D_{KL}(\mathcal Q \,\|\, \mathcal P)}{\partial w_i^{\mu,a}}
    = \beta \left[
    \mathbb E_{\mathcal{Q}}\left( \frac{\partial \mathcal H^{\text{RBM}}(\boldsymbol{\sigma},
    \boldsymbol{z} \mid \boldsymbol{w})}{\partial w_i^{\mu,a}}\right)
    -
    \mathbb E_{\mathcal{P}} \left(\frac{\partial \mathcal H^{\text{RBM}}(\boldsymbol{\sigma},
    \boldsymbol{z} \mid \boldsymbol{w})}{\partial w_i^{\mu,a}}\right)
    \right],
\end{equation}
where $\E_{\mathcal{Q}}$ is the average over the empirical
distribution $\mathcal Q$ and $\E_{\mathcal{P}}$ is the average
over the equilibrium distribution $\mathcal P$ of the modular RBM at temperature
$\beta^{-1}$. Since the cost function~\eqref{eq:Hamiltonian_RBMM} is linear in the weights,
\begin{equation}\label{eq:dH_dw}
    \frac{\partial \mathcal H^{\text{RBM}}(\boldsymbol{\sigma}, \boldsymbol{z}
    \mid \boldsymbol{w})}{\partial w_i^{\mu,a}}
    = -\frac{1}{\sqrt{N}}\, \sigma_i^a z_\mu^a ,
\end{equation}
so that, denoting by $\epsilon$ the learning rate, the gradient-descent learning
rule~\eqref{eq:KL_grad} reads as
\begin{equation}\label{eq:learning_rule_w}
    \Delta w_i^{\mu,a} = \dfrac{\epsilon \beta}{\sqrt N} \left(
    \E_{\mathcal{Q}} (\sigma_i^a z_\mu^a)
    - \E_{\mathcal{P}} (\sigma_i^a z_\mu^a )
    \right), \quad \forall i, \mu, a.
\end{equation}
As standard, since the equilibrium measure $\mathcal P$ is intractable for large sizes, we resort instead to Markov Chain Monte Carlo sampling via alternating Gibbs sweeps. Specifically, running the chain requires the two conditional distributions, derived in
Appendix~\ref{App:A}, which read
\begin{align}
     \mathcal{P}(\bm z\mid \bm \sigma)
     &= \prod_{\mu=1}^K\dfrac{\beta^{L/2}}{(2\pi)^{L/2}\sqrt{\det\bm g}}
        \exp\!\Big[-\tfrac{1}{2}\big(\bm z_\mu - \bm g\,\bm u_\mu(\bm\sigma)\big)^{\!T}
        (\beta^{-1}\bm g)^{-1}\big(\bm z_\mu - \bm g\,\bm u_\mu(\bm\sigma)\big)\Big],
     \label{eq:Pzgivensigma}\\[4pt]
     \mathcal{P}(\bm \sigma\mid \bm z)
     &= \prod_{i=1}^N\prod_{a=1}^L
        \dfrac{\exp\!\big(\beta \sigma_i^a h_i^a(\bm z)\big)}{2\cosh\!\big(\beta h_i^a(\bm z)\big)},
     \label{eq:Psigmagivenz}
\end{align}
where we introduced the two local fields
\begin{equation}\label{eq:u_and_h}
    u^a_\mu (\bm \sigma)= \dfrac{1}{\sqrt{N}} \SOMMA{i=1}{N} w_i^{\mu,a} \sigma_i^a,
    \qquad
    h_i^a (\bm z) = \dfrac{1}{\sqrt{N}} \SOMMA{\mu=1}{K}w_i^{\mu,a}z_{\mu}^a .
\end{equation}
As expected, setting $L=1$ recovers the conditionals of the single
RBM~\cite{BARRA20121,LEONELLI2021314}.

\begin{remark}
\label{rem:clamp}
Equation \eqref{eq:Pzgivensigma} explains the encoding \eqref{eq:clamp} of the
label. Writing the local field \eqref{eq:u_and_h} in terms of the generalised
Mattis magnetisation
\begin{equation}
    \label{eq:Mattis_magnetisation}
    m_\mu^a(\bm\sigma|\bm w):=\dfrac1N\sum_{i=1}^N w_i^{\mu,a}\sigma_i^a
\end{equation}
one has $u_\mu^a(\bm\sigma)=\sqrt N\,m_\mu^a(\bm\sigma| \bm w)$; since
\eqref{eq:Pzgivensigma} is Gaussian with mean $\bm g \bm u_{\mu}(\bm \sigma)$, a
sampled hidden configuration reads
\begin{equation}\label{eq:z_scale}
    z_\mu^a=\sqrt N\,\bigl(\bm g\,\bm m_\mu(\bm\sigma| \bm w)\bigr)^a+O_{\mathcal P}(\beta^{-1/2}) .
\end{equation}
Hidden activations thus inherit the scale of the Mattis magnetisation: an $O(1)$
label would force $m_\mu^a = O(N^{-1/2})$, i.e.\ no macroscopic overlap and, through
\eqref{eq:u_and_h}, a vanishing local field $h_i^a$. Labels of order
$\sqrt N$ are therefore required. Their {shape} follows by evaluating
\eqref{eq:z_scale} at the dual configuration of Section~\ref{sec:Model} (weights set
to the archetypes, each module aligned with its own, so that $m_\mu^b=\delta_{\mu,\mu_b}$),
which returns precisely \eqref{eq:clamp}: module $a$ must excite its own target
hidden neuron while suppressing, at $-\lambda$, the neurons claimed by the other
modules, precisely what the inter-module inhibition $\bm x = \bm g^{-1}$ demands.
\end{remark}

In practice, both the \emph{positive phase} $\E_{\mathcal{Q}} (\sigma_i^a z_\mu^a)$ and \emph{the negative phase} $\E_{\mathcal{P}} (\sigma_i^a z_\mu^a )$
appearing in \eqref{eq:learning_rule_w} are obtained as approximations, as detailed below, see also Algorithm~\ref{alg:cd1}.

Assuming that no module is more likely than another to carry a given label (so that the position a label occupies inside a tuple is uninformative) and that the noise affecting examples is site and label independent, we can outline a ``\emph{mini-batch}'' $\mathcal B_\mu$ containing $|\mathcal B_\mu|=B$ items of the training dataset, such that $\mu\in\hat{\bm\mu}^{\,p}$. Since the labels within an item are distinct, each of them contributes exactly one column carrying the label $\mu$, that is one example $\bm\eta^{\mu,m}$ of the archetype $\bm\xi^\mu$. Thus, we define the co-called ``cue''
\begin{equation}\label{eq:majority_cue}
    s^{\mu}_i(\mathcal B_\mu) := \mathrm{sign}\Bigl(\SOMMA{m\in\mathcal B_{\mu}}{}{\eta}^{\mu,m}_i  \Bigr)\in\{-1,+1\};
    \qquad i=1,\dots,N.
\end{equation}
this operation is implementable due to the layer-wise supervised nature of the dataset. Note also that $\bm s^{\mu}$ depends on the archetype alone and not on the tuple in which it is presented, so that the same cue serves whichever companions $\mu$ happens to be drawn with. Then, we can estimate the positive phase as
\begin{equation} \label{eq:clamped}
	\widehat{\E}_{\mathcal Q}\bigl(\sigma_i^a z_\mu^a\bigr)
	= \
\begin{cases}
\ \ \ s_i^{\mu_a} \sqrt{N}, & \textrm{if} ~ \mu=\mu_a,\\
-\lambda\, s_i^{\mu_a} \sqrt{N}, & \textrm{if} ~ \mu \neq \mu_a, ~ \mu = \mu_b,\\
\ \ \ 0, & \textrm{otherwise},
\end{cases}
\end{equation}
for $a,b=1,\dots,L$.

As for the negative phase, this is estimated by a single contrastive-divergence step (CD-1): starting from a clamped data
configuration, the
visible layer is reconstructed from~\eqref{eq:Psigmagivenz}, and the hidden layer
is sampled from ~\eqref{eq:Pzgivensigma}. The clamped configuration is taken according to \eqref{eq:clamped}, namely
\begin{equation} \label{eq:initial}
\sigma_i^a = s_i^{\mu_a}, \qquad
z_{\mu}^a=
\begin{cases}
\ \ \ \sqrt{N}, & \textrm{if} ~ \mu=\mu_a,\\
-\lambda \sqrt{N}, & \textrm{if} ~ \mu \neq \mu_a, ~ \mu = \mu_b,\\
\ \ \ 0, & \textrm{otherwise}.
\end{cases}
\end{equation}

\subsection{Fixed point theorem} \label{sec:fixed_point_thm}
We now ask whether the setting \eqref{eq:setting}, $\bm w=\bar{\bm\eta}$, is
stationary under the update \eqref{eq:learning_rule_w}. 
Before stating the result, it is worth recalling that, despite being defined in terms of the averages $\E_{\mathcal{Q}}$ and $\E_{\mathcal{P}}$,  $\Delta \bm w$ remains a random variable, since a single CD-1 step involves several independent sources of randomness. Four of them are internal to the algorithm: the item drawn from $\mathcal S$, which fixes the $L$-tuple $\hat{\bm\mu}$ of labels to be classified, the mini-batches $\bm{\mathcal B}=(\mathcal B_{\mu_1},\dots,\mathcal B_{\mu_L})$, one per module, from which the positive phase is built, the uniform variables $\bm \omega$ used to sample the reconstructed visible layer, and the Gaussian noise $\bm \zeta$ of the negative phase, see Algorithm \ref{alg:cd1}. The fifth is the generation of the training set itself that is the corruption variables $\bm \chi$ of eq. \eqref{eq:esempi}, which enter the update through the empirical means $\bar{\bm\eta}$. Accordingly we write 
\begin{equation}\label{eq:averages}
    \E_{\mathrm{alg}}[\,\cdot\,] := \E\bigl[\,\cdot \mid \mathcal S, \bm\xi\,\bigr],
    \qquad
    \E_{\mathcal S}[\,\cdot\,] \ \ \text{(average over $\bm\chi$, at fixed $\bm\xi$)},
    \qquad
    \E_{\bm\xi}[\,\cdot\,] \ \ \text{(average over the archetypes)},
\end{equation}
and we reserve the plain symbol
\begin{equation}\label{eq:E_def}
    \E[\,\cdot\,] := \E_{\mathcal S}\,\E_{\mathrm{alg}}[\,\cdot\,]
\end{equation}
 for the average of the update with respect to the law of the data at fixed archetypes. The archetypes are always kept quenched: they are latent variables of the process generating the dataset and, we stress, they are never accessed by the algorithm, which only uses the examples and their labels. The fixed-point property we are defining hereafter is a property of the average \eqref{eq:E_def}, that is, it states that the learning rule does not carry any drift.
\begin{algorithm}[t]
\DontPrintSemicolon
\caption{Supervised CD-1 step for the Modular RBM}
\label{alg:cd1}
\KwIn{training set $\mathcal S=\bigl\{(\hat{\bm\eta}^{\,p},\hat{\bm z}^{\,p})\bigr\}_{p=1,\dots,P}$; coupling $\bm g$; inverse temperature $\beta$;
mini-batch size $B$ odd; learning rate $\epsilon$ (with $\epsilon\beta\ll1$)}
\KwOut{weight update $\Delta w^{\mu,a}_i$}
\BlankLine
\tcc{Initialisation}
$w^{\mu,a}_i \leftarrow \bar\eta^{\mu}_i=\tfrac{1}{rM}\sum_{m=1}^{M}\eta^{\mu,m}_i$
\tcp*{shared across $a$; or $w^{\mu,a}_i\sim\mathcal N(0,1)$}
Draw an item $p$ from $\mathcal S$ (eq.~\ref{eq:Q_def}) and read its label tuple
$\hat{\bm\mu}=\hat{\bm\mu}^{\,p}=(\mu_1,\dots,\mu_L)$\;
\BlankLine
\tcc{Positive phase: both layers are clamped to the datum}
\For{$a=1,\dots,L$}{
  Draw a mini-batch $\mathcal B_{\mu_a}$ of $B$ items of $\mathcal S$ bearing $\mu_a$, $B$ odd
  \tcp*{one column of label $\mu_a$ per item}
  $(\sigma^a_i)' \leftarrow s^{\mu_a}_i(\mathcal B_{\mu_a})=\mathrm{sign}\bigl(\sum_{m\in\mathcal B_{\mu_a}}\eta^{\mu_a,m}_i\bigr)$
  \tcp*{majority-vote cue, eq.~\eqref{eq:majority_cue}}
  $\hat z^{a}_\mu \leftarrow \sqrt N\bigl(\delta_{\mu\mu_a}
     -\lambda\sum_{c\neq a}\delta_{\mu\mu_c}\bigr)$
  \tcp*{supervised clamp, eq.~\eqref{eq:clamp}}
}
\BlankLine
\tcc{Reconstruction: sample visible, $\mathcal P(\bm\sigma\mid\bm z)$ (eq.~\ref{eq:Psigmagivenz})}
$h^a_i \leftarrow \dfrac{1}{\sqrt N}\sum_{\mu=1}^{K}w^{\mu,a}_i\hat z^a_\mu
   = \sum_{b=1}^{L} g_{ab}\,w^{\mu_b,a}_i$
\tcp*{only the $L$ labelled patterns contribute}
$(\sigma_i^a)'' \leftarrow \mathrm{sign}\!\big([\,1+e^{-2\beta\,h^a_i}\,]^{-1}
   - \omega_i^a\big)$,\quad $\omega_i^a\sim\mathcal U(0,1)$\;
\BlankLine
\tcc{Negative phase: resample hidden from the reconstruction, $\mathcal P(\bm z\mid\bm\sigma)$ (eq.~\ref{eq:Pzgivensigma})}
$(z_\mu^a)'' \leftarrow \displaystyle\sum_{b=1}^{L} g_{ab}\,u^b_\mu(\bm\sigma'') + (\zeta_\mu^a)''$,\quad
   $\bm\zeta_\mu\sim\mathcal N(\bm0,\beta^{-1}\bm g)$\;
\BlankLine
\tcc{Weight update}
$\Delta w_i^{\mu,a} \leftarrow \dfrac{\epsilon\beta}{\sqrt N}
   \big[(\sigma_i^a)'\,\hat z_\mu^a - (\sigma_i^a)''(z_\mu^a)''\big]$\;
\end{algorithm}

\begin{definition}[Mean fixed point of CD-1]
    \label{def:fixed_point}
    Consider Algorithm~\ref{alg:cd1}, in which the visible and the hidden layers of module $a$ are clamped as specified in \eqref{eq:initial}. Denote by $\Delta w^{\mu,a}_{i}$ the resulting update. A weight configuration $\bm w^{\star}$ is a \emph{fixed point in mean} of CD-1 if, setting $\bm w = \bm w^{\star}$,
    \[
    \E\bigl[\Delta w^{\mu,a}_{i}\bigr] = 0
    \qquad \text{for every } i = 1,\dots,N, \quad \mu = 1,\dots,K, \quad a = 1,\dots,L,
    \]
    the expectation being the one defined in \eqref{eq:E_def}.
\end{definition}
 We now state the fixed-point Theorem, while we refer to the Appendix \ref{app:fixed_point} for the full proof.

\begin{theorem}
    \label{thm:fixed_point}
    Consider the modular RBM trained by the supervised CD-1 step of Algorithm~\ref{alg:cd1}, with weights planted at $\bm w = \bar{\bm\eta}$, the visible layers $\boldsymbol \sigma^a$ clamped to a cue of the archetype $\mu_a$ built by majority vote over a ``batch'' of size $B = |\mathcal B_{\mu_a}|$, and the hidden layers clamped according to \eqref{eq:clamp}. Let $\rho$ be the dataset entropy \eqref{eq:rho_def} and set
    \[
      c_1:=\frac{(1-r^2)\bigl[1-\lambda(L-1)\bigr]^2}{8\bigl[1+\lambda(L-1)\bigr]^2},
      \qquad
      C:=\frac{20\bigl[1+\lambda(L-1)\bigr]}{r},
      \qquad
      C':=1+\lambda(L-1).
    \]
    Then, for every $N\ge4$, every odd $ B\ge1$ such that each label is assigned to at least $B$ items of $\mathcal{S}$, every $\beta>0$ and every $\gamma>0$, with probability at least $1-2KL\,e^{-N\gamma^{2}/2}$ over the realisation of the archetypes,
    \begin{equation}\label{eq:thm_bound}
      \max_{i,\mu,a}\bigl|\E[\Delta w^{\mu,a}_i]\bigr|
      \;\le\;
      \epsilon\beta\left[
      C\Bigl(
      \underbrace{e^{-\beta(1-\lambda(L-1))}}_{\text{thermal noise}}
      +\underbrace{L\,e^{-c_{1}/\rho}}_{\text{dataset noise}}
      \Bigr)
      \;+\;\underbrace{2e^{-Br^{2}/2}}_{\text{mini-batch noise}}
      \;+\;\underbrace{C'\gamma}_{\text{finite size}}\right] .
    \end{equation}
\end{theorem}

\begin{corollary}
    \label{cor:fixed_point}
    Choosing $\gamma=\sqrt{2\ln(2KLN)/N}$ in \eqref{eq:thm_bound}, the bound
    \[
      \max_{i,\mu,a}\bigl|\E[\Delta w^{\mu,a}_i]\bigr|
      \;\le\;
      \epsilon\beta\left[C\Bigl(e^{-\beta(1-\lambda(L-1))}+L\,e^{-c_{1}/\rho}\Bigr)
      +2e^{-Br^{2}/2}
      +C'\sqrt{\frac{2\ln(2KLN)}{N}}\right]
    \]
    holds with probability at least $1-N^{-1}$ over the archetypes. Consequently
    $\E[\Delta w^{\mu,a}_i]\to0$ as $\beta\to\infty$, $M\to\infty$, $B\to\infty$ and $N\to\infty$ with $\ln K = o(N)$, in any order, so that $\bar{\bm\eta}$ is a fixed point in mean of supervised CD-1 in the sense of Definition~\ref{def:fixed_point}. Averaging over the archetypes rather than conditioning on them gives, for every fixed $i,\mu,a$, the same bound with the last term replaced by $C'/\sqrt N$.
\end{corollary}
The bound \eqref{eq:thm_bound} splits the residual drift into four independent contributions, whose detailed derivations is given in Appendix~\ref{App:B}.
\newline
The first contribution represents the probability that the reconstruction step flips a visible spin under the conditional distribution \eqref{eq:Psigmagivenz}. It is controlled by the margin of the local field \eqref{eq:u_and_h}, which quantifies how strongly the field points towards the target archetype. Inter-module inhibition plays a crucial role here: on sites where different modules claim overlapping archetypes, the local field is reduced from $1$ to $1-\lambda(L-1)$. The margin therefore remains strictly positive when $\lambda(L-1)<1$, a condition that simultaneously ensures the positive definiteness of $\bm g$ and the stationarity of the learning rule.
\newline
The second contribution is the accuracy of the weights. It is the probability that the empirical estimate of an archetype is so noisy at a given site that the field there points the wrong way. It depends on the data only through the entropy $\rho$, and it is exponentially small as soon as $Mr^2\gg1$, that is as soon as the examples are numerous enough for their average to resolve the archetype.
\newline
The third contribution accounts for the accuracy of the batch presented to the visible layer. Since the input is constructed via a majority vote over $B$ noisy examples (see Eq. \eqref{eq:initial}), it generally deviates from the exact archetype. This residual mismatch is exponentially small in $Br^2$, becoming negligible for a sufficiently large batch size.
\newline
The fourth contribution is the overlap that two distinct archetypes have by accident at finite $N$. It is a property of the ground-truths that happen to be drawn rather than of the algorithm: the typical overlap between two of them is of order $N^{-1/2}$, which is the rate at fixed $(i,\mu,a)$, while holding uniformly over all the $KLN$ triples costs a further logarithm, whence the $\sqrt{2\ln(2KLN)/N}$ of Corollary~\ref{cor:fixed_point}. It is the only contribution that survives the limits $\beta\to\infty$, $M\to\infty$ and $B\to\infty$.

\begin{remark}
The theorem holds for every $L$, including $L=1$, i.e.\ a single RBM, and it generalises the result of~\cite{LEONELLI2021314} in three respects. First, it extends the analysis from a single machine to $L$ coupled modules, the inhibition entering through the combination $\lambda(L-1)$ discussed above. Second, it upgrades the statement from a purely asymptotic one to a quantitative bound \eqref{eq:thm_bound}. Third, and independently of $L$, it uses the supervised hypothesis in full: the label is not an $O(1)$ perturbation of a sampled hidden state but the hidden state itself, at its own scale \eqref{eq:clamp}.
\end{remark}

\section{Numerical results} \label{sec:Numerical}
We now turn to the numerical side. We first check Theorem~\ref{thm:fixed_point} directly, measuring the residual drift of the learning rule at the planted weights and next we verify that $\bm{\bar \eta}$ is also an attractor of the CD-1 dynamics, estimating its basin of attraction. We then verify that those weights actually solve the two tasks the machine was designed for, namely simple classification (case 2a) and hard classification (case 2b).
\subsection{Numerical verification of Thm \ref{thm:fixed_point}.}
\label{sec:num_fixed_point}
\begin{figure}[t!]
    \centering
    \begin{subfigure}[b]{0.325\textwidth}
        \centering
        \includegraphics[width=\textwidth]{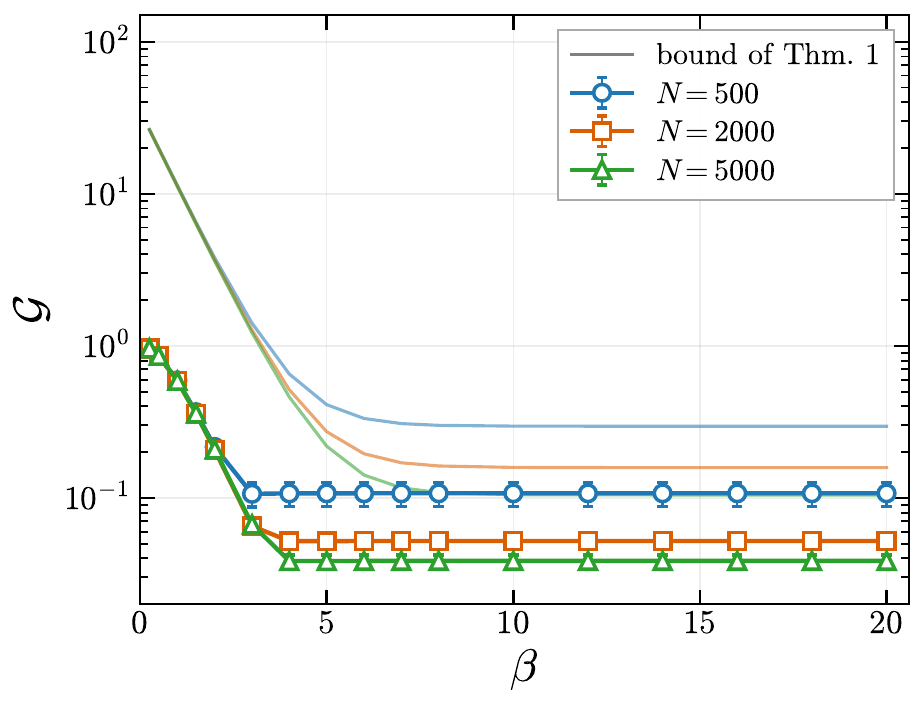}
        \caption{}
        \label{fig:panel_thermal}
    \end{subfigure}
    \hfill
    \begin{subfigure}[b]{0.325\textwidth}
        \centering
        \includegraphics[width=\textwidth]{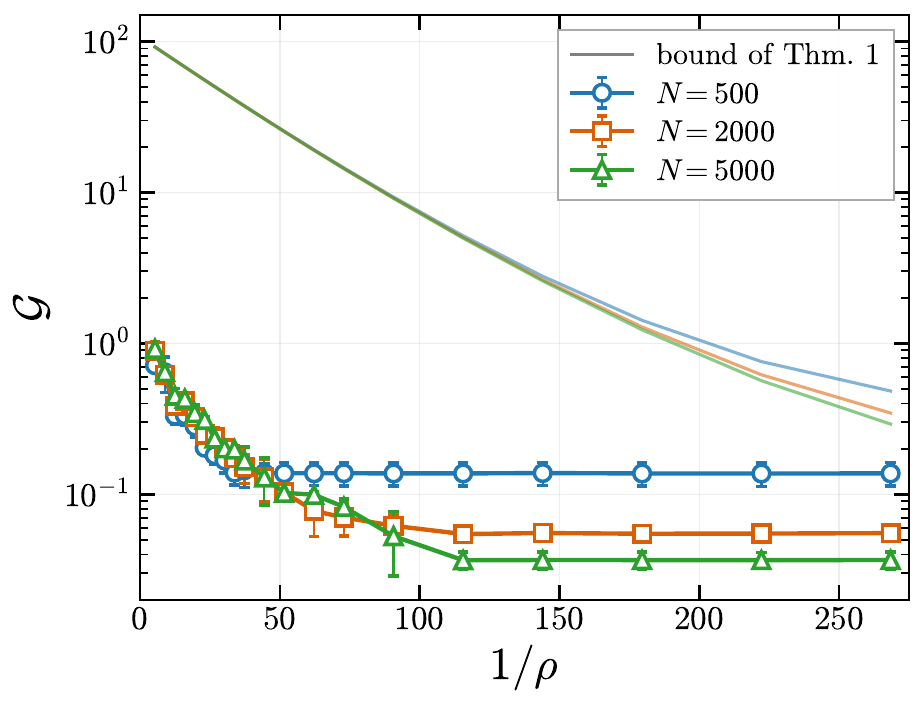}
        \caption{}
        \label{fig:panel_dataset}
    \end{subfigure}
    \hfill
    \begin{subfigure}[b]{0.325\textwidth}
        \centering
        \includegraphics[width=\textwidth]{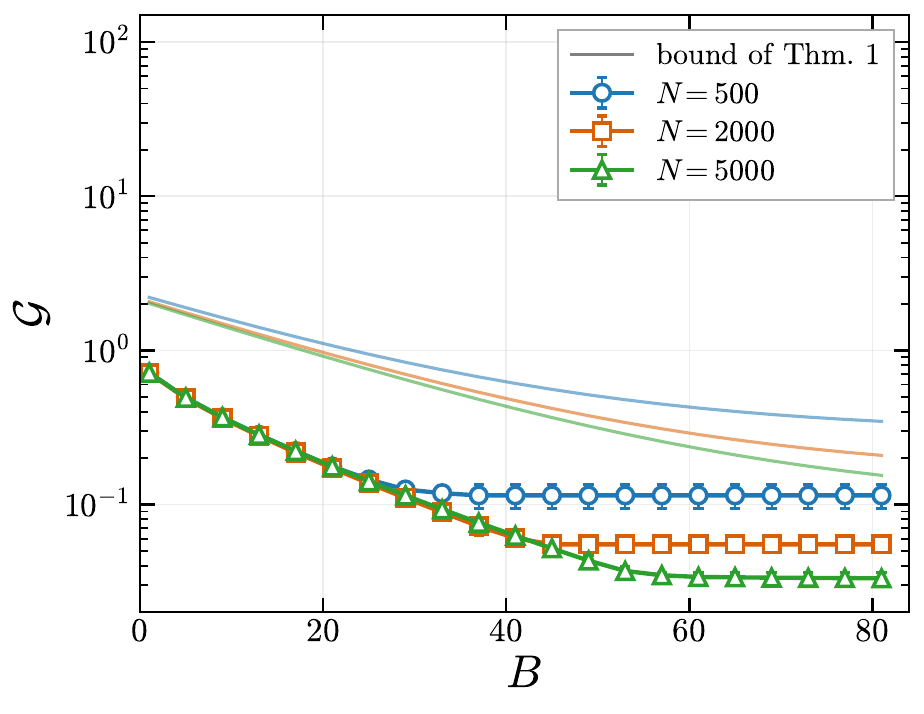}
        \caption{}
        \label{fig:panel_minibatch}
    \end{subfigure}
    \caption{Numerical verification of Theorem~\ref{thm:fixed_point}. Residual drift $\mathcal G$ of eq.~\eqref{eq:Ginf}, evaluated at the planted weights $\bm w = \bar{\bm\eta}$ for a modular RBM with $L=3$ modules, $K=20$ archetypes and inter-module coupling $\lambda = 0.2$. The average over the generation of the dataset is estimated on $20$ independent realisations, and the curves report the typical value of $\mathcal G$ over $4$ independent samples of the quenched archetypes, the error bars being the spread over the latter. The thin solid lines are the right-hand side of \eqref{eq:thm_bound} evaluated with its own constants, one per system size and in the corresponding colour.
    \emph{(a)} Thermal channel, at $r=0.8$, $M=10^{5}$ and $B=25$, so that both the dataset and the mini-batch channels are switched off: the residual decays exponentially in $\beta$, until it saturates on a plateau.
    \emph{(b)} Dataset channel, at $r=0.8$ and $\beta=40$: the residual decays exponentially in $1/\rho$ and reaches the same plateau.
    \emph{(c)} Mini-batch channel, at $r=0.3$, $M=5001$ and $\beta=40$: the residual decays exponentially in $B$ and stays below the bound . In all three panels the plateau is the finite-size channel of \eqref{eq:thm_bound}, which no other parameter can remove and which orders the three curves as $N^{-1/2}$.}
    \label{fig:fixed_point_numerics}
\end{figure}

 \begin{figure}[t!]
    \centering
    \begin{subfigure}[b]{0.43\linewidth}
        \centering
        \includegraphics[width=\linewidth]{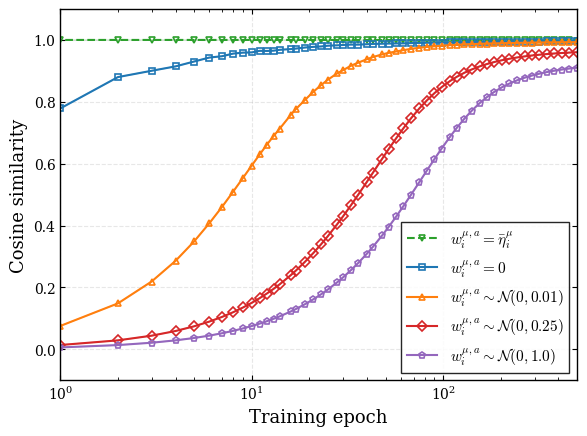}
        \caption{}
        \label{fig:init_comparison_a}
    \end{subfigure}
    \hfill
    \begin{subfigure}[b]{0.43\linewidth}
        \centering
        \includegraphics[width=\linewidth]{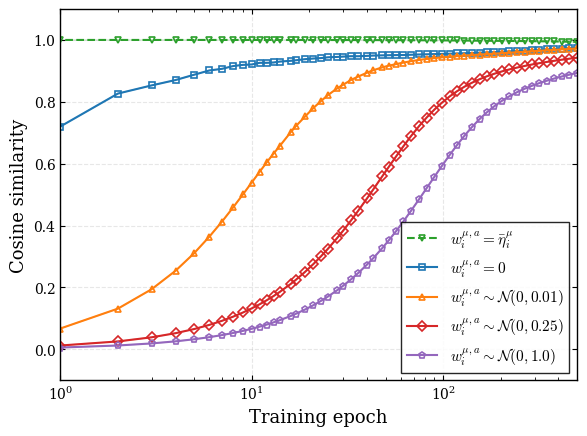}
        \caption{}
        \label{fig:init_comparison_b}
    \end{subfigure}
    \caption{Cosine similarity between the learned weights $\bm w(\textrm{epoch})$ and the initialization based on the empirical mean of the examples $\bar{\bm \eta}$, evaluated during training in a modular RBM. The network consists of $L=3$ modules, each with $N=5000$ neurons and inter-module coupling $\lambda=0.2$, storing $K=20$ archetypes from which $M=50$ examples per archetype are drawn with an entropy of $\rho=0.01$ (left) and $\rho= 0.2$ (right), batch size $B=25$. Moreover the system evolves with inverse temperature $\beta=5$. The training is performed using Algorithm \ref{alg:cd1} for $500$ epochs and with learning rate $0.03$. The green line tracks the self similarity of a network initialized on the mean of the examples $w_i^{\mu,a}=\bar{\eta}_i^\mu$. The other colored lines track the cosine similarity evolution of networks with different initialization: one with zero weights and the others with a random Gaussian with increasing variance. Regardless of the initial condition, the learning dynamics successfully drives the weights toward the initialization based on the mean of the examples. The time to reach such state scales with the variance of the initial noise.}
    \label{fig:init_comparison}
\end{figure}

Theorem~\ref{thm:fixed_point} bounds the mean update at $\bm w = \bar{\bm\eta}$, so the natural observable is the update \eqref{eq:learning_rule_w} stripped of the learning rate and of the temperature:
\begin{equation}\label{eq:Ginf}
    \mathcal G := \max_{i,\mu,a}\;
    \frac{\bigl|\E\bigl[\Delta w^{\mu,a}_i\bigr]\bigr|}{\epsilon\beta} .
\end{equation}
This is exactly the quantity that the left-hand side of \eqref{eq:thm_bound} controls: $\mathcal G \simeq 0$ certifies that $\bar{\bm\eta}$ is a genuine fixed point of the dynamics, while a large $\mathcal G$ signals that the very first gradient step would drive the weights away from it. Dividing by $\epsilon\beta$ is essential, since otherwise the update could be made arbitrarily small simply by lowering the learning rate.
At fixed dataset the three averages internal to a CD-1 step, over the mini-batch, over the sampling of the reconstructed visible layer and over the Gaussian noise of the negative phase, can all be carried out in closed form, as shown in Appendix~\ref{App:B}. The two phases are therefore free of Monte Carlo error, and the only sampling left is over the generation of the dataset and over the archetypes\footnote{We have
checked this reduction against a direct simulation of Algorithm~\ref{alg:cd1}, obtaining a Pearson correlation of $0.9999$ over all
triples $(i,\mu,a)$.}. 

Each panel of Fig.~\ref{fig:fixed_point_numerics} switches off every contributions but one. In panel~(a) the dataset is essentially noiseless and the mini-batch large, so that only the thermal channel survives: $\mathcal G$ falls exponentially in $\beta$ and then saturates.  
Panel~(b) repeats the exercise on the dataset channel, at low temperature and with the mini-batch taken to be the whole pool of examples, so that the cue is as clean as the data allow. Here $\mathcal G$ falls exponentially in $1/\rho$. 
Panel (c) probes the term opened by the mini-batch, the other two sources of noise being switched off,
so that the residual drift is set entirely by the quality of the batch. It decays exponentially in
$B$ and stays below the Theorem bound. 
Finally, the plateau reached in each panels recovers the fourth channel of \eqref{eq:thm_bound}, the one that neither $\beta$ nor $M$, nor $B$ can remove, scaling as $N^{-1/2}$. The three panels therefore confirm the four terms of \eqref{eq:thm_bound} separately and with
the predicted functional form: $\bar{\bm\eta}$ is a fixed point of CD-1 in the sense of
Definition~\ref{def:fixed_point}, up to a residual drift that is exponentially small in $\beta$,
in $1/\rho$ and in $B$, and of order $N^{-1/2}$ in the size of the network.

However, Theorem~\ref{thm:fixed_point} says nothing about whether the dynamics reaches the
configuration $\boldsymbol w = \bar{\boldsymbol \eta}$, nor about what happens when it is started elsewhere.
Figure~\ref{fig:init_comparison} settles both points, by running CD-1 at a finite
learning rate from five different initial conditions, and at two values of the dataset
entropy, following the cosine similarity between the weights and $\bar{\bm\eta}$.
Planted at $\bar{\bm\eta}$, the weights stay there for the whole run, so the
stationarity of Theorem~\ref{thm:fixed_point}, which holds in mean, survives the
fluctuations of a single CD-1 step at a finite learning rate. Started instead from configurations that have no
overlap with it, either from zero weights or from independent Gaussian entries, the
dynamics can still converge to that same point. The learning rule therefore
rebuilds the empirical mean of the examples from scratch, with a transient that
lengthens monotonically with the variance of the initial noise, from a few epochs
when the weights start at zero, where there is no noise to undo and only signal to
build, to several hundred at unit variance, where the initial weights have a norm
comparable to $\bar{\bm\eta}$ itself and are orthogonal to it. Raising the entropy of
the dataset from $\rho=0.01$ to $\rho=0.2$ slows the approach and slightly lowers the
level at which the curves settle, as the dataset channel of \eqref{eq:thm_bound}
prescribes, without altering the picture.

\subsection{Simple classification}\label{sec:Pattern_recognition}
  
\begin{figure}[t!]
\centering
\begin{subfigure}[b]{0.48\textwidth}
	\centering
	\includegraphics[width=\textwidth]{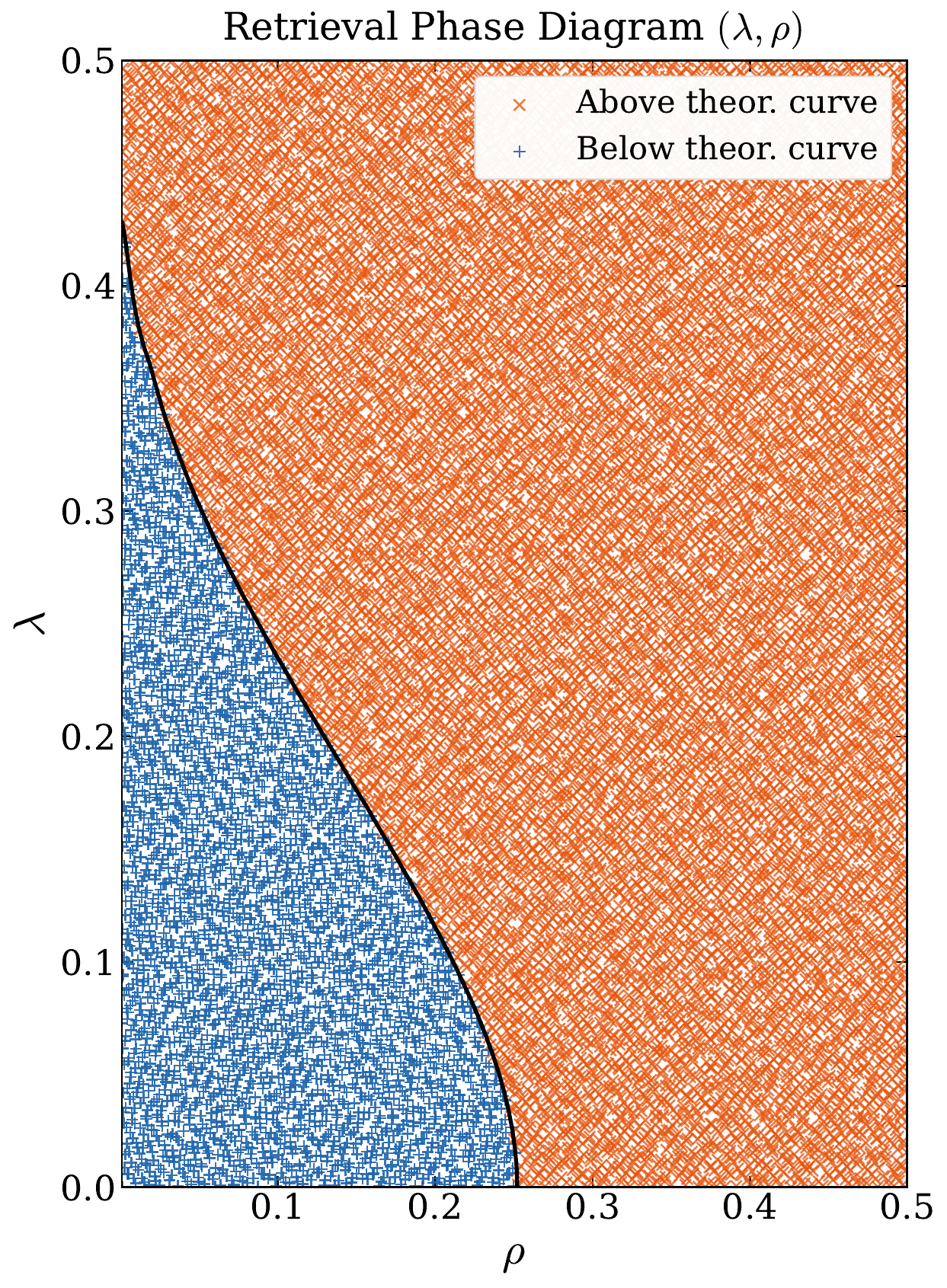}
	\caption{}
	\label{fig:stability_modular}
\end{subfigure}
\hfill
\begin{subfigure}[b]{0.48\textwidth}
	\centering
	\includegraphics[width=\textwidth]{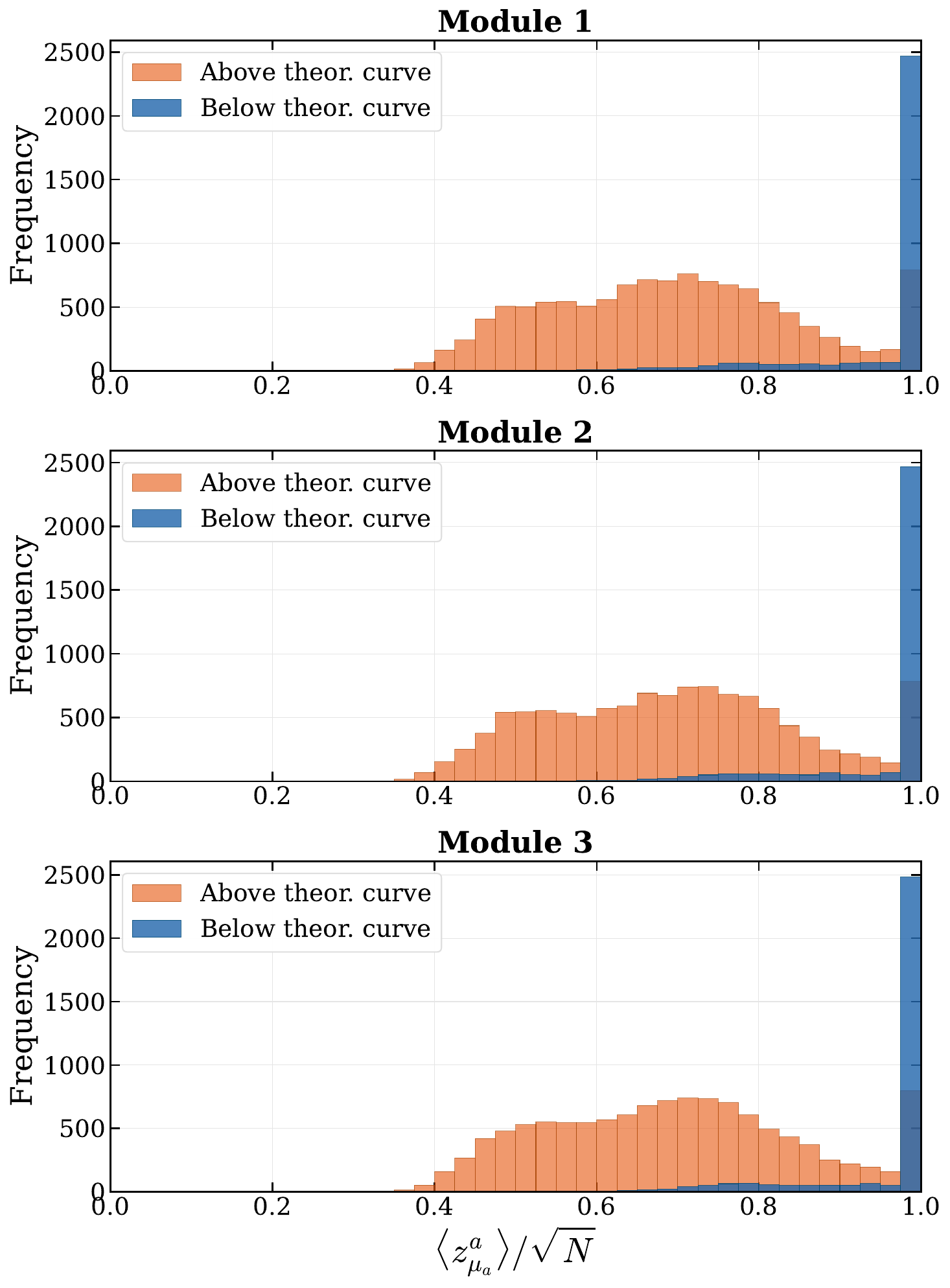}
	\caption{}
	\label{fig:histogram_retrieval}
\end{subfigure}
    \caption{Classification performance of a modular RBM
    in the $(\lambda, \rho)$ plane.
    	The network consists of $L = 3$ modules, each with $N = 5000$ neurons,
    	storing $K = 50$ archetypes from which $M = 50$ examples per archetype are
    	drawn. The system evolves at inverse temperature $\beta = 10$ for $30$ steps and each module being initialized with a test example $\hat{\boldsymbol \eta}_{\textrm{test}}$. Results are averaged over
    	$N_{\mathrm{runs}} = 50$ independent realisations per parameter point,
    	sampled via a $15000$-point Sobol sequence.
    	\textit{Left panel}: the solid black curve is the theoretical phase
    	boundary $\langle m \rangle = 0.95$ obtained from the fixed-point analysis
    	of the supervised modular HM model~\cite{AGLIARI2026131134}. Each marker
    	corresponds to a sampled $(\lambda, \rho)$ point, distinguished by whether
    	it lies below the theoretical curve (blue pluses) or above it (orange
    	crosses).
    	\textit{Right panel}: normalised histogram of the hidden-unit activations
    	$\langle z^{a}_{\mu_a} \rangle_{\mathrm{}}$, shown separately for points
    	below (blue) and above (orange) the theoretical curve. Points below the
    	boundary exhibit a sharp peak near unity, confirming that each hidden module
    	responds selectively to its corresponding archetype, whereas points above
    	the boundary display a broad distribution centred at lower values,
    	signalling a failure of the classification.}
    \label{fig:phase_diagram_rho_lambda}
\end{figure}
Having established that $\bar{\bm\eta}$ is a mean fixed point of CD-1, we set
$w_i^{\mu,a}=\bar\eta_i^{\mu}$ for any $(i,\mu,a)$ and verify that these planted weights actually solve
the classification task the machine was trained for. 
Once the $L$ archetypes to be classified are fixed, the test dataset is defined as
\begin{equation}\label{eq:S_test}
\mathcal{S}_{\mathrm{test}}:= \{(\bm{\hat \eta}^{p}_{\mathrm{test}}, \bm{\hat {z}}^{p}_{\mathrm{test}})\}_{p=1, \dots, P_{\mathrm{test}}},
\end{equation}
each item following the configuration of Case 2a and drawn independently of the training set. The labels are kept out of the dynamics: only the visible layers are initialised, each on a noisy example of a distinct archetype, and the network is let evolve by alternating $\mathcal P(\bm z\mid\bm\sigma)$ (eq.~\ref{eq:Pzgivensigma}) and $\mathcal P(\bm\sigma\mid\bm z)$ (eq.~\ref{eq:Psigmagivenz}) until convergence. The labels enter only at the end, as the ground truth telling which hidden unit each module was supposed to switch on.

The outcome is read there. By \eqref{eq:z_scale}, the activation of the unit carrying the label of module $a$ is, up to fluctuations of order $\beta^{-1/2}$,
\begin{equation}\label{eq:decision}
    \frac{z^{a}_{\mu_a}}{\sqrt N}\;=\;\SOMMA{b=1}{L}g_{ab}\,m^{b}_{\mu_a}(\bm\sigma| \bar{\bm \eta}) ,
\end{equation}
which equals one when every module has retrieved its own archetype, since then $m^{b}_{\mu_a}=\delta_{ab}$, and is therefore the very value that the clamp \eqref{eq:clamp} prescribes (see Remark~\ref{rem:clamp}). We monitor it as the decision variable of module $a$: a value close to one certifies that the item presented to it has been classified correctly, while a value appreciably below one signals that the corresponding unit has failed to switch on.\footnote{Note that this is a stricter criterion than reading the prediction as $\arg\max_{\mu} z^{a}_{\mu}$, since it demands that the correct unit not merely dominate but reach the level prescribed during training.}

The simulations are compared with the analytical results obtained in~\cite{AGLIARI2026131134}\footnote{The comparison requires the inverse temperature to be matched between the two conventions. The weights used here are the unbiased estimators \eqref{eq:etabar_def}, whose squared norm satisfies $\frac1N||\bar{\bm\eta}^\mu||^2\to1+\rho$, whereas the self-consistency equations of \cite{AGLIARI2026131134} are written for unit-norm couplings. The two Hamiltonians therefore coincide once the theoretical curves are evaluated at $\beta(1+\rho)$, which is how all the boundaries reported here are obtained.}. In particular, in the low-load regime considered here the expectation of the Mattis magnetization with respect to the Boltzmann-Gibbs distribution satisfies a self-consistency equation whose solution provides the boundary
reported in panel (a) of Fig.~\ref{fig:phase_diagram_rho_lambda}, namely, it corresponds to the level curve
of the Mattis magnetization to be equal to 0.95: below that curve the model is expected to work correctly. 

In particular, in panel (a), for each point in the diagram $(\rho, \lambda)$, we simulate a modular RBM, where we assign an element of $\mathcal S_{\mathrm{test}}$ and evaluate the correct classification by computing \eqref{eq:decision}, averaged over $N_{\mathrm{runs}}$ independent runs. Specifically, the blue and orange points in the diagram distinguish between points below the theoretical curve and points outside the theoretical curve. We therefore expect the blue points to assign a value of 1 to $\langle z^{a}_{\mu_a}\rangle/\sqrt N$, while the orange points should yield a value less than 1. Panel (b) confirms our expectations. On the one hand, we have thus verified that we can consider our neural network, trained with this initialization, capable of classifying archetypes; on the other hand, it also numerically confirms the equivalence between the supervised modular HM \cite{AGLIARI2026131134} and our Boltzmann machine.

\subsection{Hard classification} \label{sec:Dinsentanglement}
\begin{figure}[t!]
\centering
	\begin{subfigure}[b]{0.32\textwidth}
		\centering
		\includegraphics[height=4.3cm, width=\textwidth, keepaspectratio]{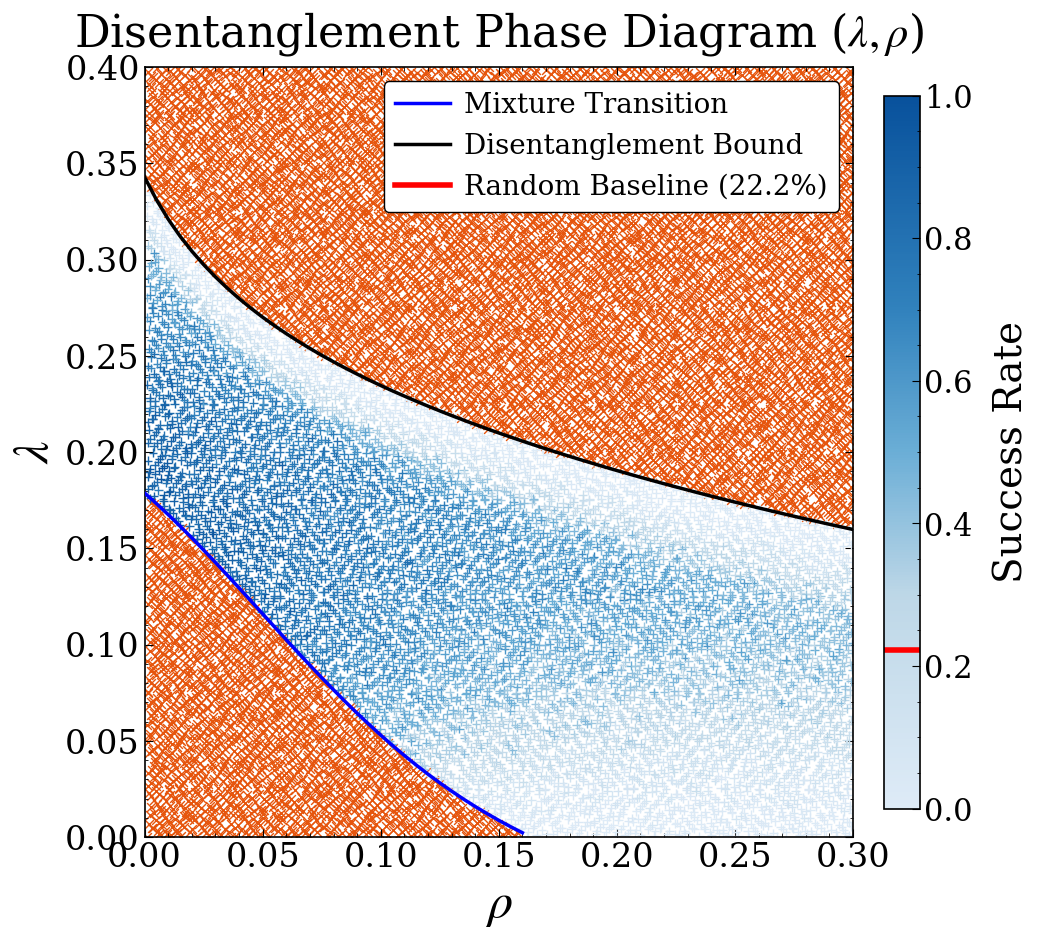}
		\caption{}
		\label{fig:phase3}
	\end{subfigure}
	\hfill
	\begin{subfigure}[b]{0.32\textwidth}
		\centering
		\includegraphics[height=4.3cm, width=\textwidth, keepaspectratio]{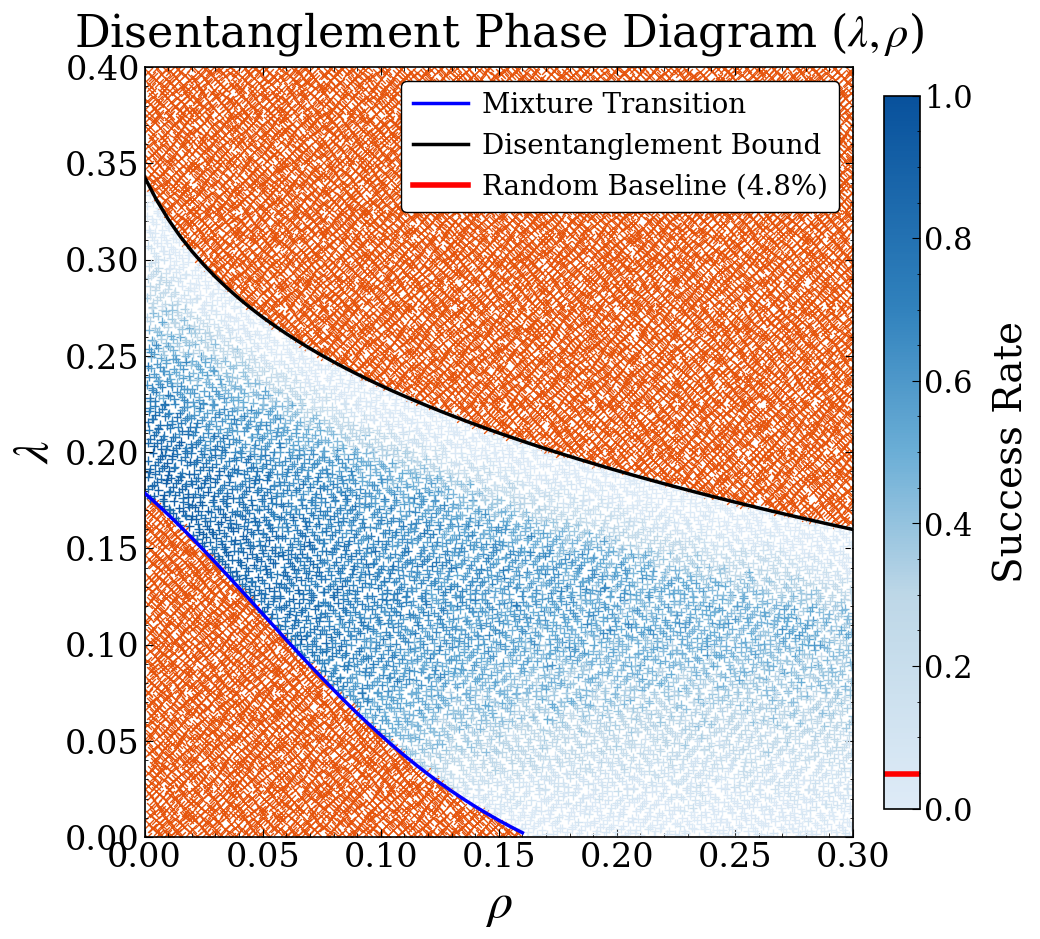}
		\caption{}
		\label{fig:phase5}
	\end{subfigure}
	\hfill
	\begin{subfigure}[b]{0.32\textwidth}
		\centering
		\includegraphics[height=4.3cm, width=\textwidth, keepaspectratio]{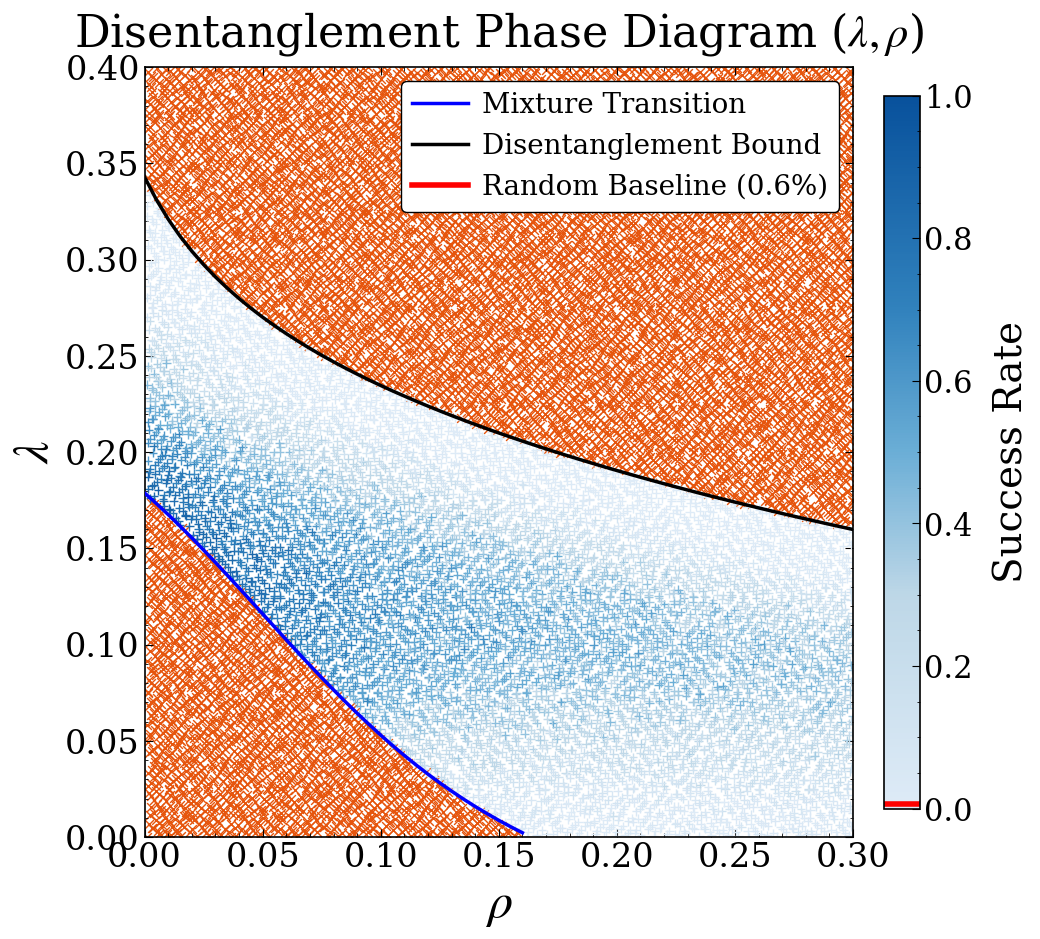}
		\caption{}
		\label{fig:phase10}
	\end{subfigure}

	\vspace{0.8em}

	\begin{subfigure}[b]{0.32\textwidth}
		\centering
		\includegraphics[height=4.3cm, width=\textwidth, keepaspectratio]{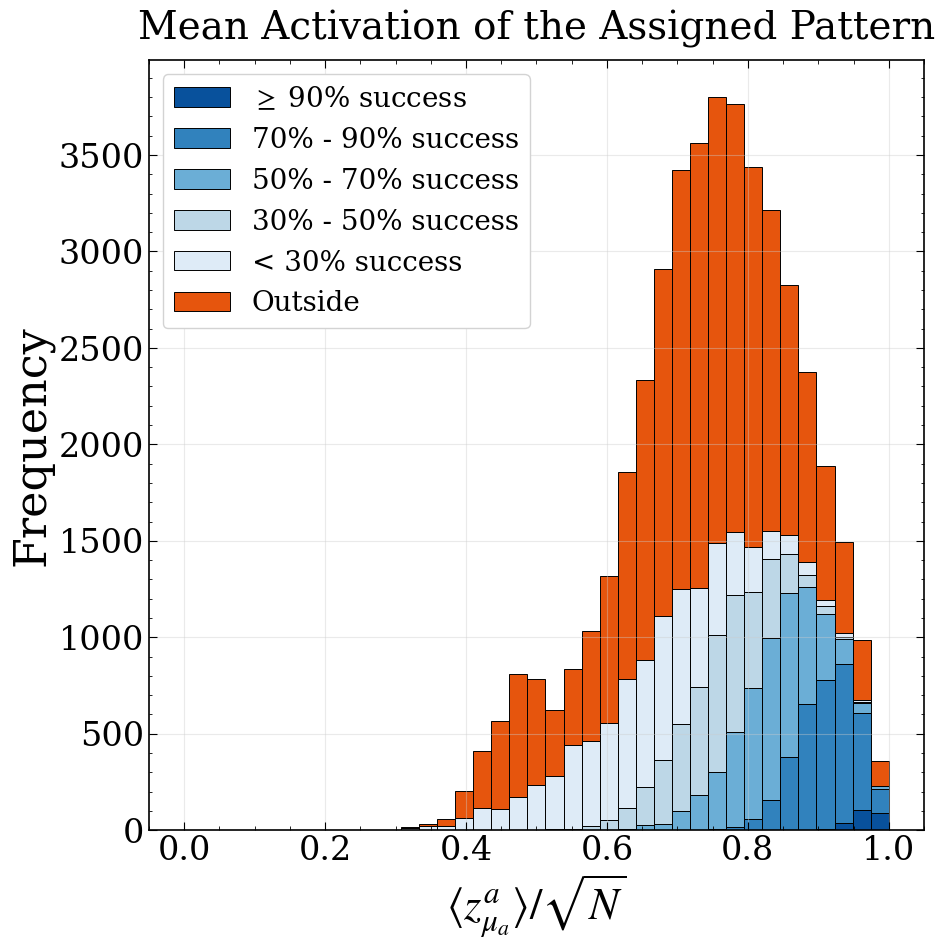}
		\caption{}
		\label{fig:isto3}
	\end{subfigure}
	\hfill
	\begin{subfigure}[b]{0.32\textwidth}
		\centering
		\includegraphics[height=4.3cm, width=\textwidth, keepaspectratio]{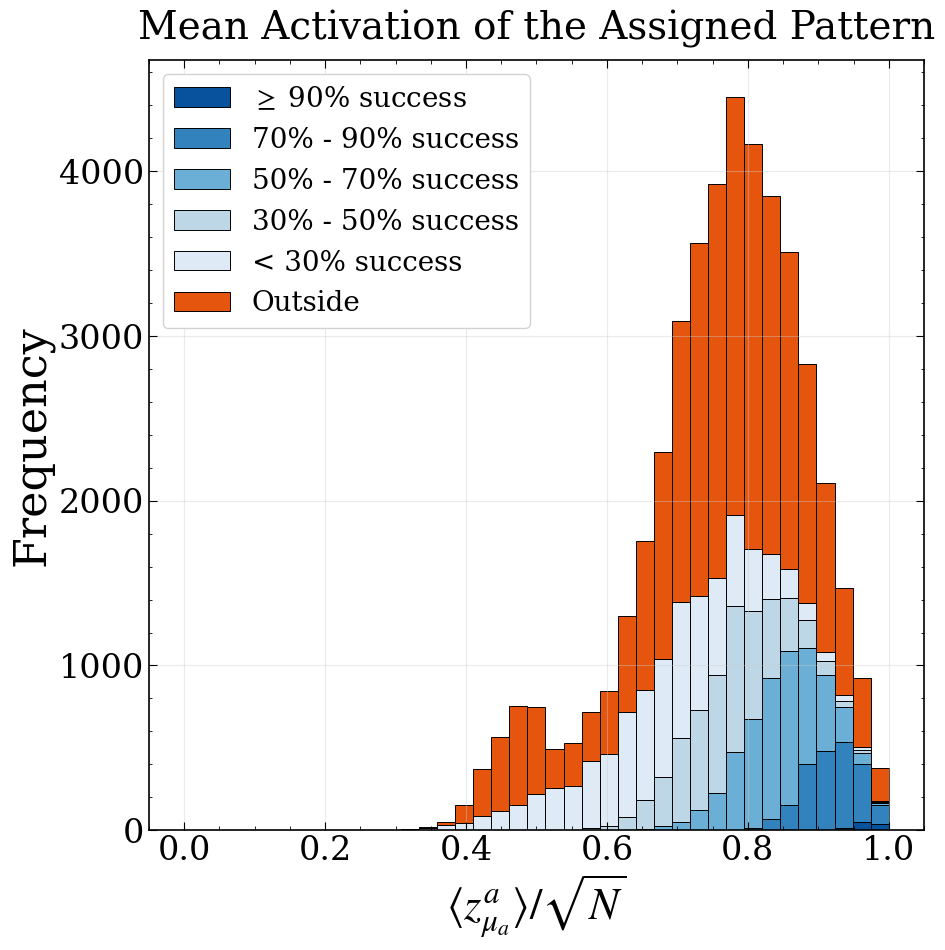}
		\caption{}
		\label{fig:isto5}
	\end{subfigure}
	\hfill
	\begin{subfigure}[b]{0.32\textwidth}
		\centering
		\includegraphics[height=4.3cm, width=\textwidth, keepaspectratio]{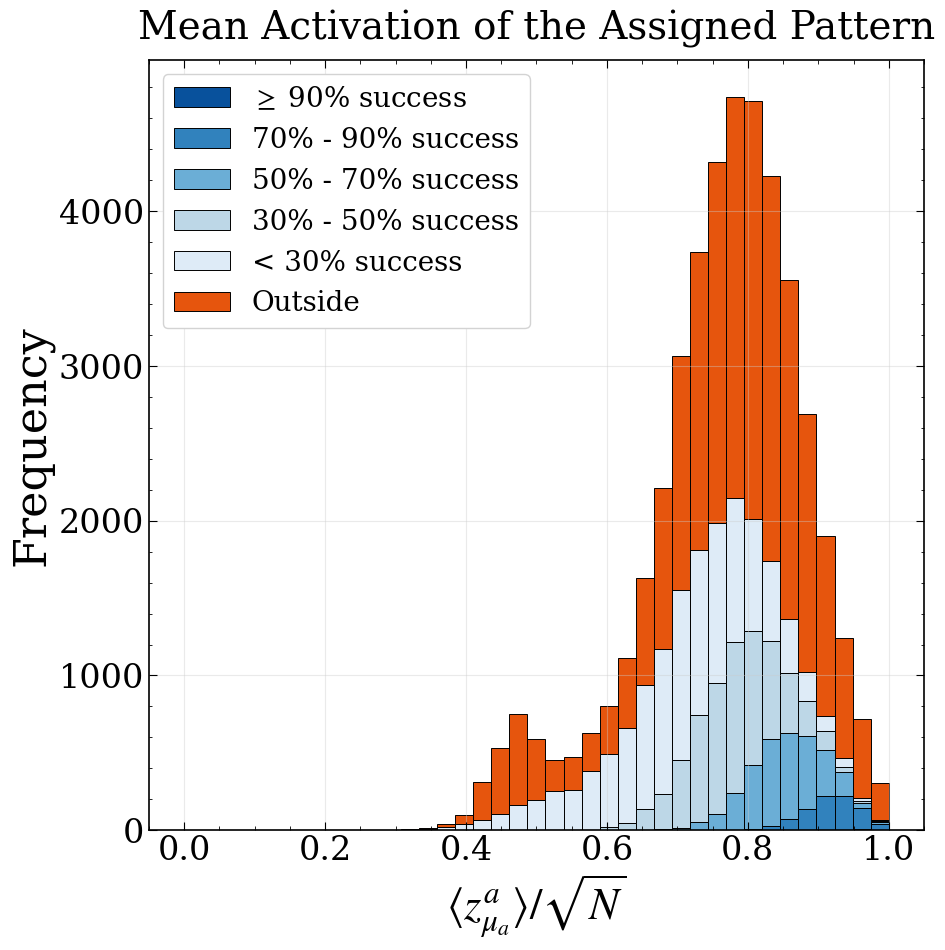}
		\caption{}
		\label{fig:isto10}
	\end{subfigure}
    \caption{Disentangling performance of a modular RBM in the ($\rho,\lambda$) plane. The network consists of $L=3$ modules, each with $N=5000$ neurons, storing $K$ archetypes from which $M=50$ examples per archetype are drawn. The system evolves at inverse temperature $\beta=5$ for 75 steps and each module is initialized with the same mixture $\hat{\boldsymbol \eta}_{\textrm{test}}^{\prime}$ obtained from examples of different archetypes. Results are averaged over $N_{\mathrm{runs}} = 50$ independent realizations per parameter point, sampled via a 15000-point Sobol sequence. The panels show results for $K=3$ (a,d), $K=5$ (b,e) and $K=10$ (c,f). \textit{Top panels (a,b,c)}: solid curves represent theoretical boundary for different system phases \cite{AGLIARI2026131134}. Each marker of the plot corresponds to a samples $(\rho,\lambda)$ point, distinguished by the fact that it is within the theoretical curves or outside. For the point inside we use a continuous blue colormap based on the success rate of the runs performed. Here, the successful threshold is set as $\theta =0.9$. The solid red line on the colorbar indicates the random baseline success rate, that is the probability $L!/K^L$ of assigning the $L$ correct archetypes to the $L$ modules when each module picks uniformly among the $K$ stored archetypes. \textit{Bottom panels (d,e,f)}: histogram of the hidden unit mean activation of the assigned pattern. Point with a high success rate (darker blue) exhibit a sharp peak at values near 1, confirming that each module successfully retrieved different pattern from the mixture. Lighter blue points exhibit peaks for lower value which occur when the disentangling task failed. Orange point represent the outside region where disentangling is unsuccessful. 
    }
\label{fig:phase_diagram_rho_lambda_dis}
\end{figure}
We now fix the internal weights of the machine $w_i^{\mu,a} = \bar \eta_i^\mu$, for every $i, \mu,a$ and see if this initialization might be meaningful even for the harder task. In particular, we define the following test dataset:

\begin{equation}
    \label{eq:harder_test_dataset}
    \mathcal{S}'_{\mathrm{test}}:= \{((\hat{\bm \eta}^p_{\mathrm{test}})', (\bm{\hat z}^p_{\mathrm{test}})')\}_{p=1,\dots,P'_{\mathrm{test}}},
\end{equation}
each item following the configuration 2b of Section \ref{sec:Model}. As in \ref{sec:Pattern_recognition}, we evaluate the outcome on the hidden layer and we compare our results with the theory of the supervised modular network \cite{AGLIARI2026131134}. 
In particular, the boundaries reported in the top panel of Fig.~\ref{fig:phase_diagram_rho_lambda_dis}
enclose the region where the mixture configurations (e.g., $\boldsymbol \sigma = \hat{\boldsymbol \eta}_{\textrm{test}}$) is unstable while the disentangled state (e.g., $\boldsymbol \sigma = (\boldsymbol \xi^{\mu_1},...,\boldsymbol \xi^{\mu_L} )$) is not, and this is
precisely the window in which the decomposition can occur. 

Specifically, for each point $(\rho,\lambda)$ of the diagram we initialise the machine with
an element of $\mathcal S'_{\mathrm{test}}$, we let the network relax, and we call the run
successful when the $L$ modules end up on $L$ distinct archetypes. We therefore expect the
successful points to fall between the two curves, and the failures to accumulate below the blue
one, where the modules stay locked on the mixture, and above the black one, where retrieval is
destroyed. Moreover, the points inside the curves are also differentiated based on the success rate over the runs performed. Indeed, a run is considered successful if the network not only exits the mixture state but also manages to allocate a different archetype to each of the $L$ modules. Formally, this one-to-one matching is formulated as an optimal assignment problem and solved using the Hungarian algorithm. Let $\pi$ be a permutation of the $L$ archetypes composing the mixture; we use the algorithm to find the optimal assignment $\pi ^ *$ that maximizes the total activation of the hidden units across all modules: 
\begin{equation}
\pi ^ *= \arg \max_{\pi \in \mathcal{M}_L} \SOMMA{a=1}{L} \dfrac{\langle z_{\pi(a)}^a \rangle}{\sqrt{N}}
\end{equation}
where $\mathcal{M}_L$ is the set of all possible one-to-one matching between the $L$ modules and the $L$ archetypes and $\dfrac{\langle z_{\pi(a)}^a \rangle}{\sqrt{N}}$ is the activation of the specific hidden unit in module $a$ that corresponds to archetype $\pi (a)$, the brackets again indicate an average on $N_{\mathrm{runs}}$ independent runs.
After extracting the optimal matching we check if the activation of every assigned hidden unit exceeds a predefined threshold $\theta$: 
\begin{equation}
    \dfrac{\langle z_{\pi(a)}^a \rangle}{\sqrt{N}} > \theta, \quad \forall a \in \{1,...,L\}
\end{equation}
satisfying this condition confirms the simultaneous retrieval of the different archetypes composing the mixture.
The bottom panels of the figure shows the mean activation of the hidden neurons. For parameter region with high success rate we observe a sharp peak near unity. This confirms the expectation that the competitive dynamics force the network to classify different archetype for each module. On the other hand for point outside the theoretical curves or with low enough success rate we observe two different peaks reflecting two types of failing scenarios. The peak for lower activation emerges when the modules fail to break the symmetry and remains trapped in the mixture state. The second peak at high values emerges in the trial in which the network manages to converge to the archetypes but fails to retrieve all of them; in this case multiple modules retrieve the same archetype failing to find the whole set composing the mixture, and causing the one-to-one check to fail.

\section{Conclusions}
We asked whether the weight configuration suggested by the HM--RBM duality
is one that learning actually reaches, or only one a theorist can plant by
hand. For simple classification the answer is a theorem: the empirical-mean
weights $\bar{\bm\eta}$ are a fixed point in mean of one-step contrastive
divergence, for any number of modules $L$, with a residual drift that splits
into different contributions, ultimately controlled by the training-dataset size and quality, and by the network architecture. For the harder classification task, $\bar{\bm\eta}$ still works
well across the region predicted by the
theory developed for the corresponding HM~\cite{AGLIARI2026131134}.
This kind of result should matter beyond this model wherever a closed-form
initialisation replaces trial and error, especially where weights are
physical and slow to reprogramme, as in memristor-based Boltzmann
machines~\cite{Ernoult-SciRep2019}. 

Extending the analytical results to the hard classification task and to structured archetypes are two natural next steps. However, probably more interesting is investigating whether $\bm x$ needs to be fixed at all. If contrastive divergence,
started from a random $\bm x$, converges on its own to the antagonistic
coupling we impose by hand, units would learn not just to classify but to
compete. Hebb's postulate is that cells firing together wire together; we
leave open whether cells asked to fire apart learn to stay apart.

\acknowledgments
The authors are grateful to Paulo Duarte Mourão for useful discussions. \\
E.A. acknowledges Sapienza Università di Roma (RM124190CB1269EB). \\
The authors acknowledge the use of the Lagrange Multi-GPU Server at the Department of Mathematics, Sapienza University of Rome, for computational resources supporting this work.

\appendix
\section{Derivation of the Conditional Distributions}\label{App:A}

In this appendix we derive explicitly the conditional distributions $\mathcal P(\bm z|\bm \sigma)$, $\mathcal P(\bm \sigma|\bm z)$, and the marginal $\mathcal P(\bm \sigma)$ for the Modular RBM defined by the cost function
\begin{equation}
    \mathcal{H}^{\mathrm{RBM}}(\bm \sigma,\bm z \mid \bm w) = -\frac{1}{\sqrt{N}} \sum_{a=1}^{L} \sum_{\mu=1}^{K} \sum_{i=1}^{N} w_i^{\mu,a} \sigma_i^a z_\mu^a + \frac{1}{2} \sum_{\mu=1}^{K} \sum_{a,b=1}^{L} z_\mu^a (g^{-1})_{ab} z_\mu^b.
\end{equation}
The joint distribution is given by the Boltzmann--Gibbs measure
\begin{equation}
    \mathcal P(\bm \sigma, \bm z) = \frac{1}{\mathcal Z^{\mathrm{RBM}}(\bm w)} e^{-\beta \mathcal{H}^{\mathrm{RBM}}(\bm \sigma, \bm z \mid \bm w)},
    \qquad
    \mathcal Z^{\mathrm{RBM}}(\bm w) = \sum_{\bm \sigma} \int {d}z\, e^{-\beta \mathcal{H}^{\mathrm{RBM}}(\bm \sigma, \bm z \mid \bm w)}.
\end{equation}
We remind the shorthand
\begin{equation}
    u_\mu^a(\bm \sigma) = \frac{1}{\sqrt{N}} \sum_{i=1}^{N} w_i^{\mu,a} \sigma_i^a,
    \qquad
    \boldsymbol{u}_\mu(\bm \sigma) := \bigl(u_\mu^1(\bm \sigma), \ldots, u_\mu^L(\bm \sigma)\bigr) \in \mathbb{R}^L,
\end{equation}
so that the Hamiltonian can be written compactly as
\begin{equation}
    \mathcal{H}^{\mathrm{RBM}}(\bm \sigma, \bm z \mid \bm w) = \sum_{\mu=1}^{K} \left[ \frac{1}{2} \bm z_\mu^{\mathsf{T}} \bm g^{-1} \bm z_\mu - \bm z_\mu^{\mathsf{T}} \boldsymbol{u}_\mu(\bm \sigma) \right].
\end{equation}

Let us start by computing $\mathcal P(\bm z| \bm \sigma)$. 

Since the exponent factorises over $\mu$, we have:
\begin{equation}
\mathcal P(\bm z|\bm \sigma) = \prod_{\mu=1}^{K} \mathcal P(z_\mu|\bm \sigma). 
\end{equation}
Fixing $\mu$, the unnormalised density reads
\begin{equation}
    \mathcal P(z_\mu|\bm \sigma) \;\propto\; \exp\!\left\{-\beta \!\left[\frac{1}{2} \bm z_\mu^{\mathsf{T}} \bm g^{-1} \bm z_\mu - \bm z_\mu^{\mathsf{T}} \boldsymbol{u}_\mu(\bm \sigma)\right]\right\}.
\end{equation}
Completing the square in $z_\mu$,
\begin{equation}
    \frac{1}{2} \bm z_\mu^{\mathsf{T}} \bm g^{-1} \bm z_\mu - \bm z_\mu^{\mathsf{T}} \boldsymbol{u}_\mu(\bm \sigma)
    = \frac{1}{2}\bigl(\bm z_\mu - \bm g\,\boldsymbol{u}_\mu(\bm \sigma)\bigr)^{\mathsf{T}} \bm g^{-1} \bigl(\bm z_\mu - \bm g\,\boldsymbol{u}_\mu(\bm \sigma)\bigr)
      - \frac{1}{2}\boldsymbol{u}_\mu(\bm \sigma)^{\mathsf{T}} \bm g\,\boldsymbol{u}_\mu(\bm \sigma),
\end{equation}
so the normalisation integral over $\mathbb{R}^L$ is a standard Gaussian integral. One obtains
\begin{equation}
    \mathcal P(\bm z|\bm \sigma) = \prod_{\mu=1}^{K} \frac{\beta^{L/2}}{(2\pi)^{L/2}\sqrt{\det g}}
    \exp\!\left\{-\frac{1}{2}\bigl(\bm z_\mu - \bm g\,\boldsymbol{u}_\mu(\bm \sigma)\bigr)^{\mathsf{T}}(\beta^{-1}\bm g)^{-1}
    \bigl(\bm z_\mu - \bm g\,\boldsymbol{u}_\mu(\bm \sigma)\bigr)\right\},
\end{equation}
i.e.\ $\mathcal P(z_\mu|\bm \sigma) \sim \mathcal{N}\!\left(\bm g\,\boldsymbol{u}_\mu(\bm \sigma),\, \beta^{-1} \bm g\right)$.

Now, in order to compute $\mathcal P(\bm \sigma| \bm z)$ we start from
\begin{equation}
    -\beta \mathcal{H}^{\mathrm{RBM}}(\bm \sigma,\bm z\mid \bm w)
    = \frac{\beta}{\sqrt{N}} \sum_{i,\mu,a} w_i^{\mu,a} \sigma_i^a z_\mu^a
      - \frac{\beta}{2}\sum_\mu z_\mu^{\mathsf{T}} g^{-1} z_\mu,
\end{equation}
the second term does not depend on $\bm \sigma$, so we can factor it out. Reminding the local field
\begin{equation}
    h_i^a(\bm z) = \frac{1}{\sqrt{N}} \sum_{\mu=1}^{K} w_i^{\mu,a} z_\mu^a,
\end{equation}
the joint density factorises over $(i,a)$:
\begin{equation}
    \mathcal P(\bm \sigma, \bm z) \propto \prod_{i=1}^{N}\prod_{a=1}^{L} e^{\beta h_i^a(z)\,\sigma_i^a}
    \cdot \exp\!\left\{-\frac{\beta}{2}\sum_\mu \bm z_\mu^{\mathsf{T}} \bm g^{-1} \bm z_\mu\right\}.
\end{equation}
Summing over all $\bm \sigma \in \{-1,+1\}^{NL}$ to obtain the normalisation $\mathcal Z^{\mathrm{RBM}}(\bm z)$ gives
\begin{equation}
    \mathcal Z^{\mathrm{RBM}}(\bm z) = \prod_{i=1}^{N}\prod_{a=1}^{L} 2\cosh\!\bigl(\beta h_i^a(z)\bigr)
    \cdot \exp\!\left\{-\frac{\beta}{2}\sum_\mu \bm z_\mu^{\mathsf{T}} \bm g^{-1} \bm z_\mu\right\},
\end{equation}
and hence
\begin{equation}
\mathcal P(\bm \sigma|\bm z) = \prod_{i=1}^{N}\prod_{a=1}^{L}
    \frac{e^{\beta \sigma_i^a h_i^a(z)}}{2\cosh\!\bigl(\beta h_i^a(z)\bigr)}.
\end{equation}

Finally, integrating out the hidden variables,in particular the normalisation
\begin{equation}
    \int_{\mathbb{R}^{KL}} \prod_\mu \mathrm{d}^L z_\mu\, e^{-\beta \mathcal H^{\mathrm{RBM}}}
    = \left(\frac{2\pi}{\beta}\right)^{KL/2} (\det \bm g)^{K/2}
      \exp\!\left\{\frac{\beta}{2}\sum_{\mu=1}^{K} \boldsymbol{u}_\mu(\bm \sigma)^{\mathsf{T}} \bm g\,\boldsymbol{u}_\mu(\bm\sigma)\right\},
\end{equation}
one obtains
\begin{equation}
    \mathcal P(\bm \sigma) \;\propto\;
    \exp\!\left\{\frac{\beta}{2N}\sum_{\mu=1}^{K}\sum_{a,b=1}^{L} g_{ab}
    \sum_{i,j=1}^{N} w_i^{\mu,a} w_j^{\mu,b} \sigma_i^a \sigma_j^b \right\}.
\end{equation}
Setting $w_i^{\mu,a} = \xi_i^\mu$ recovers exactly the Boltzmann--Gibbs weight of the LAM model (eq.~\eqref{eq:LAM_hamiltonian}), thereby confirming the equivalence established in Section~\ref{sec:Model}.
\section{Proofs}\label{App:B}
\label{app:fixed_point}

In this appendix we prove Theorem~\ref{thm:fixed_point} and Corollary~\ref{cor:fixed_point}, namely that initialising the weights of the modular RBM at the empirical averages of the training examples constitutes a fixed point in mean of the one-step contrastive divergence. Throughout, the averages are those introduced in \eqref{eq:averages}--\eqref{eq:E_def}: $\E_{\mathrm{alg}}$ denotes the average over the randomness internal to Algorithm~\ref{alg:cd1} at fixed dataset and fixed archetypes, $\E_{\mathcal S}$ the average over the generation of the dataset at fixed archetypes, $\E = \E_{\mathcal S}\E_{\mathrm{alg}}$ their composition, and $\E_{\bm\xi}$ the average over the archetypes, which play the role of quenched disorder.

To lighten the notation we abbreviate
\begin{equation}\label{eq:B_shorthand}
    \Lambda := \lambda(L-1)\in[0,1),
    \qquad
    \hat g^{a}_{\mu} := \SOMMA{b=1}{L} g_{ab}\,\delta_{\mu,\mu_b} = \dfrac{\hat z^a_\mu}{\sqrt N},
    \qquad
    \Delta^{\mu}_{i} := \bar\eta^{\mu}_{i} - \xi^{\mu}_{i},
\end{equation}
so that $\SOMMA{b=1}{L}|g_{ab}| = 1+\Lambda$ for every $a$, while $|\hat g^a_\mu|\le1$ and $\hat g^a_\mu = 0$ unless $\mu$ is one of the $L$ presented labels. Note that, $\bar{\bm\eta}$ being the unbiased estimator \eqref{eq:etabar_def}, the deviation $\Delta^\mu_i$ is now measured from the archetype itself rather than from a shrunken copy of it, and that $|\bar\eta^\mu_i|\le r^{-1}$ deterministically, since $\bar\eta^\mu_i$ is an average of $M$ signs divided by $r$. We further set
\begin{equation}\label{eq:B_shorthand2}
    h^{a}_{i} := \SOMMA{b=1}{L} g_{ab}\,w^{\mu_b,a}_{i},
    \qquad
    t^{a}_{i} := \tanh\bigl(\beta h^{a}_{i}\bigr),
    \qquad
    T^{b}_{\mu} := \dfrac1N \SOMMA{j=1}{N} w^{\mu,b}_{j}\,t^{b}_{j},
    \qquad
    e^{a}_{i} := t^{a}_{i}-\xi^{\mu_a}_{i},
\end{equation}
and we work with the prefactor-free update
\begin{equation}\label{eq:B_Dy}
    \Delta y^{\mu,a}_{i} := \dfrac{\Delta w^{\mu,a}_{i}}{\epsilon\beta}
    = \dfrac{1}{\sqrt N}\Bigl[(\sigma^a_i)'\,\hat z^{a}_{\mu} - (\sigma^a_i)''\,(z^a_\mu)''\Bigr],
\end{equation}
in terms of which the statement of Theorem~\ref{thm:fixed_point} concerns $\E[\Delta y^{\mu,a}_i]$. We finally recall that, by \eqref{eq:esempi}, $\E_{\mathcal S}[\eta^{\mu,m}_i] = r\,\xi^{\mu}_i$, whence $\E_{\mathcal S}[\bar\eta^\mu_i]=\xi^\mu_i$ and $\E_{\mathcal S}[\Delta^\mu_i]=0$.

We begin by collecting the auxiliary probabilistic tools required in the proof, then state and prove the main result.

\begin{lemma}[Hoeffding's inequality]
\label{lem:hoeffding}
Let $X_1,\dots,X_n$ be independent random variables satisfying $\mathcal P(a_k\le X_k\le b_k)=1$ for every $k$, and set $S=\SOMMA{k=1}{n}X_k$. Then, for every $t>0$,
\[
  \mathcal P\bigl(S-\E[S]\le -t\bigr)
  \;\le\;
  \exp\!\left(-\frac{2t^2}{\SOMMA{k=1}{n}(b_k-a_k)^2}\right),
  \qquad
  \mathcal P\bigl(|S-\E[S]|\ge t\bigr)
  \;\le\;
  2\exp\!\left(-\frac{2t^2}{\SOMMA{k=1}{n}(b_k-a_k)^2}\right),
\]
the second bound following from the first and from its mirror image applied to $-S$.
\end{lemma}

\begin{lemma}\label{lem:emp_mean_dev}
Let $\Delta^\mu_i=\bar\eta^\mu_i-\xi^\mu_i$ be the deviation of the unbiased empirical mean \eqref{eq:etabar_def} from the archetype it estimates. Then $\E_{\mathcal S}[\Delta^\mu_i]=0$, $\E_{\mathcal S}[(\Delta^\mu_i)^2] = \rho$, $|\Delta^\mu_i|\le1+r^{-1}$, and for every $\tau>0$
\begin{equation}\label{eq:B_delta_tail}
  \mathcal P\bigl(|\Delta^\mu_i|\ge\tau\bigr)\ \le\ 2\exp\Bigl(-\frac{(1-r^{2})\,\tau^2}{2\rho}\Bigr).
\end{equation}
\end{lemma}

\begin{proof}
Fix $\mu$ and $i$ and condition on $\xi^\mu_i$. The $M$ random variables $X_m:=\bigl(\eta^{\mu,m}_i-r\,\xi^\mu_i\bigr)/r$ are independent, since the examples are generated independently, and are centred because $\E_{\mathcal S}[\eta_i^{\mu,m}]=r\,\xi_i^\mu$. Each of them takes exactly two values: $\eta^{\mu,m}_i=\xi^\mu_i$ with probability $\tfrac{1 + r}{2}$, giving $X_m=r^{-1}(1-r)\xi^\mu_i$, and $\eta^{\mu,m}_i=-\xi^\mu_i$ with probability $\tfrac{1-r}{2}$, giving $X_m=-r^{-1}(1+r)\xi^\mu_i$. The two values are separated by $r^{-1}|(1-r)+(1+r)|=2/r$, so that $b_m-a_m=2/r$ irrespective of the sign of $\xi^\mu_i$ and $\sum_m(b_m-a_m)^2=4M/r^2$. Since $M\Delta^\mu_i=\SOMMA{m=1}{M}X_m$, Lemma~\ref{lem:hoeffding} applied with $t=M\tau$ gives
\[
  \mathcal P\bigl(|M\Delta^\mu_i|\ge M\tau\bigr)
  \le2\exp\Bigl(-\frac{2M^2\tau^2 r^{2}}{4M}\Bigr)=2\exp\Bigl(-\frac{M r^{2}\tau^2}{2}\Bigr),
\]
which is \eqref{eq:B_delta_tail} once $Mr^2=(1-r^2)/\rho$ is used. The variance is $\mathrm{Var}[\eta^{\mu,m}_i]/(r^2M)=(1-r^2)/(r^2M)$, which equals $\rho$ by \eqref{eq:rho_def}, and the bound $|\Delta^\mu_i|\le1+r^{-1}$ is immediate from $|\bar\eta^\mu_i|\le r^{-1}$.
\end{proof}

The next lemma quantifies the quality of the batch \eqref{eq:majority_cue}. It is the counterpart, on the visible layer, of the unbiasedness of $\bar{\bm\eta}$ on the weights: the majority vote cannot remove the noise exactly, since it must return a binary vector, but it removes it exponentially fast in the size of the mini-batch.

\begin{lemma}[Majority vote]\label{lem:majority}
Let $B$ be odd, let $\mathcal B$ be a mini-batch of $B$ examples of the archetype $\bm\xi^\mu$, drawn without replacement from those available in $\mathcal S$, at least $B$ in number, and let $s^\mu_i(\mathcal B)$ be \eqref{eq:majority_cue}. Then
\begin{equation}\label{eq:B_rB}
  \E_{\mathcal S}\bigl[s^\mu_i(\mathcal B)\bigr]=r_B\,\xi^\mu_i,
  \qquad
  r_B:=2\,\mathcal P\Bigl(\mathrm{Bin}\bigl(B,\tfrac{1+r}{2}\bigr)>\tfrac B2\Bigr)-1 ,
\end{equation}
with $r_1=r$. Moreover $r_B$ is non-decreasing along the odd integers and
\begin{equation}\label{eq:B_rB_tail}
  0\ \le\ 1-r_B\ \le\ 2\,e^{-B r^{2}/2} .
\end{equation}
\end{lemma}

\begin{proof}
By \eqref{eq:esempi}, $\eta^{\mu,m}_i=\xi^\mu_i\chi^{\mu,m}_i$ with $\chi^{\mu,m}_i\in\{\pm1\}$, and since $\xi^\mu_i\in\{\pm1\}$ and the sign function is odd,
\[
  s^\mu_i(\mathcal B)=\mathrm{sign}\Bigl(\xi^\mu_i\SOMMA{m\in\mathcal B}{}\chi^{\mu,m}_i\Bigr)
  =\xi^\mu_i\,\mathrm{sign}\Bigl(\SOMMA{m\in\mathcal B}{}\chi^{\mu,m}_i\Bigr).
\]
The corruption variables are independent and identically distributed, so that a mini-batch drawn without replacement from the available examples carries, marginally, exactly $B$ independent copies of $\chi$; the law of the right-hand side therefore does not depend on which mini-batch is drawn. Let $W:=\#\{m\in\mathcal B:\chi^{\mu,m}_i=+1\}\sim\mathrm{Bin}(B,\tfrac{1+r}{2})$, so that $\sum_{m\in\mathcal B}\chi^{\mu,m}_i=2W-B$, which is never zero because $B$ is odd. Hence the sign equals $+1$ if and only if $W>B/2$, and
\[
  \E_{\mathcal S}\Bigl[\mathrm{sign}\Bigl(\SOMMA{m\in\mathcal B}{}\chi^{\mu,m}_i\Bigr)\Bigr]
  =\mathcal P\bigl(W>\tfrac B2\bigr)-\mathcal P\bigl(W<\tfrac B2\bigr)
  =2\,\mathcal P\bigl(W>\tfrac B2\bigr)-1=r_B ,
\]
which is \eqref{eq:B_rB}; for $B=1$ one has $\mathcal P(W=1)=\tfrac{1+r}{2}$ and thus $r_1=r$.

For \eqref{eq:B_rB_tail}, note that $\E[W]=B\tfrac{1+r}{2}$, so that $\{W<B/2\}\subseteq\{W-\E[W]\le-Br/2\}$. The one-sided form of Lemma~\ref{lem:hoeffding}, applied to the sum of $B$ independent Bernoulli variables of unit range, gives $\mathcal P(W-\E[W]\le-t)\le e^{-2t^2/B}$, and with $t=Br/2$ this is $e^{-Br^2/2}$. Since $1-r_B=2\,\mathcal P(W<B/2)$, the claim follows; non-negativity of $1-r_B$ is clear from $r_B\le1$.

Finally, write $S_B:=\sum_{m\le B}\chi^{\mu,m}_i$ and $p:=\tfrac{1+r}{2}\ge q:=\tfrac{1-r}{2}$. Adding two further votes changes the sign of the tally only when $|S_B|=1$, and precisely when either $S_B=-1$ and both new votes are $+1$, or $S_B=+1$ and both are $-1$; therefore
\[
  r_{B+2}-r_B=2\bigl[p^2\,\mathcal P(S_B=-1)-q^2\,\mathcal P(S_B=+1)\bigr].
\]
The two atoms satisfy $\mathcal P(S_B=+1)/\mathcal P(S_B=-1)=p/q$, since the two events require $(B\pm1)/2$ successes out of $B$ and the binomial coefficients coincide. Substituting, $r_{B+2}-r_B=2\,\mathcal P(S_B=-1)\,p\,(p-q)\ge0$, which is the announced monotonicity.
\end{proof}

It is convenient to record separately the empirical counterpart of \eqref{eq:B_rB}. Averaging the batch over the choice of the mini-batch at \emph{fixed} training set defines
\begin{equation}\label{eq:B_sbar}
  \bar s^{\mu}_{i} := \E_{\mathrm{alg}}\bigl[s^\mu_i(\mathcal B)\bigr]
  = \frac{1}{|\mathfrak B|}\SOMMA{\mathcal B\in\mathfrak B}{}\mathrm{sign}\Bigl(\SOMMA{m\in\mathcal B}{}\eta^{\mu,m}_i\Bigr),
\end{equation}
a deterministic function of $\mathcal S$ obeying $|\bar s^\mu_i|\le1$ and, by Lemma~\ref{lem:majority} and the tower property, $\E_{\mathcal S}[\bar s^\mu_i]=r_B\,\xi^\mu_i$.
We now remind that the positive phase here is clamped, so we have to deal with the negative phase of the algorithm. We then remember that we can write

\begin{equation}
    \mathcal P\bigl((\sigma^a_i)''=s\bigr)=\frac{e^{\beta s\,h^a_i}}{2\cosh\beta h^a_i},
  \quad h^a_i=\frac{1}{\sqrt N}\SOMMA{\mu=1}{K}w^{\mu,a}_i\hat z^a_\mu
        =\SOMMA{b=1}{L}g_{ab}\,w^{\mu_b,a}_i ;
  \label{eq:B_field}
\end{equation}
where we have used the supervised clamp:
\begin{equation}
\label{eq:B_clamp}
    \hat z^a_\mu = \sqrt{N}\biggl(\delta_{\mu,\mu_a} - \lambda \sum_{c \not =a}\delta_{\mu,\mu_c}\biggr)
\end{equation}
and
\begin{equation}
    (z^a_\mu)''=\SOMMA{b=1}{L}g_{ab}\,u^b_\mu(\bm\sigma'')+(\zeta^a_\mu)'',
  \quad u^b_\mu(\bm\sigma'')=\frac{1}{\sqrt N}\SOMMA{j=1}{N}w^{\mu,b}_j(\sigma^b_j)'' ,
  \label{eq:B_neg}
\end{equation}
with $\bm\zeta''_\mu\sim\mathcal N(\bm 0,\beta^{-1}\bm g)$ drawn independently of
$\bm\sigma''$.

Finally, the visible layer of module $a$ is clamped to the cue of the
archetype $\mu_a$ carried by the mini-batch $\mathcal B_a$, $(\sigma^a_i)'=s^{\mu_a}_i(\mathcal B_a)$, so that the prefactor-free
weight update reads
\begin{equation}\label{eq:B_Deltag}
  \Delta y^{\mu,a}_i:=\frac{\Delta w^{\mu,a}_i}{\epsilon\beta}
  =\frac{1}{\sqrt N}\Bigl[(\sigma^a_i)'\,\hat z^a_\mu
   -(\sigma^a_i)''\,(z^a_\mu)''\Bigr].
\end{equation}
Note that, by Lemma~\ref{lem:majority},
\begin{equation}\label{eq:B_pos_mean}
  \E\bigl[(\sigma^a_i)'\bigr]=r_B\,\xi^{\mu_a}_i ,
\end{equation}
so that the positive phase reproduces the archetype up to the factor $r_B$, which tends to one exponentially fast in $B$.

The next lemma performs in closed form all the averages internal to the algorithm. It is exact, no approximation being involved, and it replaces a term-by-term expansion of the two phases; it is also the identity on which the numerical verification of Section~\ref{sec:num_fixed_point} is based.

\begin{lemma}[One-step reduction]\label{lem:reduction}
For any weights $\bm w$, any realised training set $\mathcal S$ and every $i,\mu,a$,
\begin{equation}\label{eq:B_reduction}
  \E_{\mathrm{alg}}\bigl[\Delta y^{\mu,a}_{i}\bigr]
  = \bar s^{\mu_a}_{i}\,\hat g^{a}_{\mu}
  \;-\; t^{a}_{i}\SOMMA{b=1}{L} g_{ab}\,T^{b}_{\mu}
  \;-\; \frac{w^{\mu,a}_{i}\bigl(1-(t^{a}_{i})^{2}\bigr)}{N} ,
\end{equation}
with $\bar s^{\mu}_{i}$ the empirical majority-vote mean \eqref{eq:B_sbar}.
\end{lemma}

\begin{proof}
Consider the two phases of \eqref{eq:B_Dy} in turn.

\emph{Positive phase.} The clamp is deterministic and, by \eqref{eq:clamp} and \eqref{eq:B_shorthand}, $\hat z^a_\mu/\sqrt N = \hat g^a_\mu$, so that the first term of \eqref{eq:B_Dy} is linear in the presented cue. Averaging over the mini-batch $\mathcal B_a$ drawn from $\mathcal S$, and recalling \eqref{eq:B_sbar},
\[
  \E_{\mathrm{alg}}\bigl[(\sigma^a_i)'\bigr]
  =\frac{1}{|\mathfrak B|}\SOMMA{\mathcal B\in\mathfrak B}{}s^{\mu_a}_i(\mathcal B)=\bar s^{\mu_a}_i,
\]
which gives the first term of \eqref{eq:B_reduction}. Note that the field $h^a_i$ driving the reconstruction is built from the labels and the weights only, so it does not depend on which mini-batch was drawn, and the two phases decouple under $\E_{\mathrm{alg}}$.

\emph{Negative phase.} Inserting the negative-phase hidden state of Algorithm~\ref{alg:cd1} and collecting the two factors $N^{-1/2}$,
\[
  \frac{(\sigma^a_i)''(z^a_\mu)''}{\sqrt N}
  =\frac1N\SOMMA{b=1}{L}g_{ab}\SOMMA{j=1}{N}
     w^{\mu,b}_j\,(\sigma^a_i)''(\sigma^b_j)''
   +\frac{(\sigma^a_i)''(\zeta^a_\mu)''}{\sqrt N} .
\]
In the Gibbs chain the noise $\bm\zeta''$ is drawn \emph{after} the visible layer has been reconstructed, and independently of it; being also centred, it gives $\E_{\mathrm{alg}}[(\sigma^a_i)''(\zeta^a_\mu)'']=\E_{\mathrm{alg}}[(\sigma^a_i)'']\,\E_{\mathrm{alg}}[(\zeta^a_\mu)'']=0$.

As for the first term, by \eqref{eq:Psigmagivenz} the field $h^a_i$ is a function of the weights and of the labels only; once the dataset is fixed it is therefore a deterministic quantity, and the conditional law $\mathcal P(\bm\sigma\mid\bm z)$ factorises over the pairs $(i,a)$. The reconstructed spins are thus conditionally independent, with $\E_{\mathrm{alg}}[(\sigma^a_i)'']=t^a_i$, whence
\[
  \E_{\mathrm{alg}}\bigl[(\sigma^a_i)''(\sigma^b_j)''\bigr]
  =t^a_i\,t^b_j+\bigl[1-(t^a_i)^2\bigr]\delta_{ij}\delta_{ab},
\]
the correction being needed only for the single term $(j,b)=(i,a)$, where $\bigl[(\sigma^a_i)''\bigr]^2=1$. Substituting and using $g_{aa}=1$ gives
\[
  \E_{\mathrm{alg}}\Bigl[\frac{(\sigma^a_i)''(z^a_\mu)''}{\sqrt N}\Bigr]
  =t^{a}_{i}\SOMMA{b=1}{L}g_{ab}\,T^{b}_{\mu}
  +\frac{w^{\mu,a}_{i}\bigl(1-(t^{a}_{i})^{2}\bigr)}{N},
\]
which completes the proof.
\end{proof}

The following lemma controls the sign of the local field $h^a_i$, and hence the quality of the reconstruction. It is stated for a one-parameter family of events: the parameter $\theta$ fixes how the worst-case signal is split between the noise of the dataset and the margin left to the thermal fluctuations, and the resulting bound may be optimised over it.

\begin{lemma}\label{lem:margin}
Let $\theta\in(0,1)$ and define
\begin{equation}\label{eq:B_tau_delta}
  \tau_\theta:=\theta\,\frac{1-\Lambda}{1+\Lambda},
  \qquad
  \delta_\theta:=(1-\theta)\,(1-\Lambda),
  \qquad
  \mathcal E_i:=\Bigl\{\max_{b\le L}\bigl|\Delta^{\mu_b}_i\bigr|\le\tau_\theta\Bigr\}.
\end{equation}
Then, on the event $\mathcal E_i$ and for every module $a$,
\begin{equation}\label{eq:B_margin}
  \xi^{\mu_a}_i\,h^a_i\ \ge\ \delta_\theta\ >\ 0,
\end{equation}
so that $\mathrm{sign}(h^a_i)=\xi^{\mu_a}_i$ and $|h^a_i|\ge\delta_\theta$. Moreover
\begin{equation}\label{eq:B_pi}
  \pi_\theta:=\mathcal P\bigl(\mathcal E_i^{\,c}\bigr)\ \le\ 2L\,e^{-c_1(\theta)/\rho},
  \qquad
  c_1(\theta):=\frac{\theta^2(1-r^2)(1-\Lambda)^2}{2(1+\Lambda)^2}.
\end{equation}
\end{lemma}

\begin{proof}
By \eqref{eq:B_shorthand2}, at $\bm w = \bar{\bm\eta}$ the field of module $a$ at site $i$ is a combination of the \emph{planted} weights of the $L$ presented archetypes only,
\[
  h^a_i=\SOMMA{b=1}{L}g_{ab}\,\bar\eta^{\mu_b}_i
       =\SOMMA{b=1}{L}g_{ab}\Bigl(\xi^{\mu_b}_i+\Delta^{\mu_b}_i\Bigr).
\]
Multiplying by $\xi^{\mu_a}_i\in\{\pm1\}$ and using $(\xi^{\mu_a}_i)^2=1$, $g_{aa}=1$ and $g_{ac}=-\lambda$ for $c\neq a$,
\begin{equation}\label{eq:B_margin_expand}
  \xi^{\mu_a}_i\,h^a_i
  =\underbrace{1-\lambda\!\!\SOMMA{c\neq a}{}\!\xi^{\mu_a}_i\xi^{\mu_c}_i}_{\text{signal}}
  \;+\;\underbrace{\SOMMA{b=1}{L}g_{ab}\,\xi^{\mu_a}_i\Delta^{\mu_b}_i}_{\text{data noise}} .
\end{equation}
Let us discuss the two contributions separately.

\emph{Signal.} The $L-1$ products $\xi^{\mu_a}_i\xi^{\mu_c}_i$ take the values $\pm1$ with equal probability, so the signal in \eqref{eq:B_margin_expand} ranges between $1-\Lambda$ and $1+\Lambda$. The worst case is attained on the sites where all the products equal $+1$, i.e.\ on a positive fraction of the sites, and there the signal equals $1-\Lambda$. We stress that this is an $O(1)$ reduction of the signal, of the same order as the signal itself, and not a subleading cross-talk: the inhibitory entries $-\lambda$ of the clamp \eqref{eq:clamp} genuinely oppose the retrieval on about half of the sites. The signal stays strictly positive precisely when $\Lambda<1$, that is, precisely when $\bm g \succ 0$. It is worth noticing that the dataset quality does not appear here: the weights being unbiased, they generate a field of unit amplitude, and the whole effect of the noise of the data has been moved into the second term.

\emph{Data noise.} By the triangle inequality the second term of \eqref{eq:B_margin_expand} is bounded in modulus by $\sum_b|g_{ab}|\,|\Delta^{\mu_b}_i|\le (1+\Lambda)\max_b|\Delta^{\mu_b}_i|$, and on $\mathcal E_i$ this is at most $(1+\Lambda)\tau_\theta = \theta(1-\Lambda)$, i.e.\ a fraction $\theta$ of the worst-case signal. Combining,
\[
  \xi^{\mu_a}_i\,h^a_i\ \ge\ (1-\Lambda)-\theta(1-\Lambda)=(1-\theta)(1-\Lambda)=\delta_\theta,
\]
which is \eqref{eq:B_margin}.

Finally, the probability of the complementary event is estimated by a union bound over the $L$ presented archetypes and Lemma~\ref{lem:emp_mean_dev},
\[
  \mathcal P\bigl(\mathcal E_i^{\,c}\bigr)
  \le\SOMMA{b=1}{L}\mathcal P\bigl(|\Delta^{\mu_b}_i|>\tau_\theta\bigr)
  \le2L\,\exp\Bigl(-\frac{(1-r^2)\tau_\theta^2}{2\rho}\Bigr),
\]
and substituting $\tau_\theta$ from \eqref{eq:B_tau_delta},
\[
  \frac{(1-r^2)\tau_\theta^2}{2}
  =\frac{\theta^2(1-r^2)(1-\Lambda)^2}{2(1+\Lambda)^2}=c_1(\theta),
\]
which is \eqref{eq:B_pi}. We remark that the constant $c_1(\theta)$ is exactly the one obtained in the biased convention: the two effects of the rescaling, a signal enlarged by $r^{-1}$ and a deviation $\Delta$ enlarged by the same factor, compensate in the threshold $\tau_\theta$, as they must, since the event $\mathcal E_i$ concerns a relative accuracy and $\rho$ is the intrinsic measure of it.
\end{proof}

\begin{lemma}\label{lem:recon}
With $e^a_i=t^a_i-\xi^{\mu_a}_i$ as in \eqref{eq:B_shorthand2} one has $|e^a_i|\le2$ always, and
\begin{equation}\label{eq:B_e_bound}
  \E_{\mathcal S}\bigl[|e^a_i|\bigr]\ \le\ \varepsilon_\beta(\theta)+2\pi_\theta,
  \qquad
  \varepsilon_\beta(\theta):=2e^{-2\beta\delta_\theta}=2e^{-2(1-\theta)\beta (1-\Lambda)},
\end{equation}
with $\delta_\theta$ and $\pi_\theta$ as in Lemma~\ref{lem:margin}. Moreover $\mathcal P\bigl((\sigma^a_i)''\neq\xi^{\mu_a}_i\bigr)\le e^{-2\beta\delta_\theta}+\pi_\theta$.
\end{lemma}

\begin{proof}
On $\mathcal E_i$, Lemma~\ref{lem:margin} guarantees that the field points in the direction of the archetype and that its modulus is at least $\delta_\theta$. Since $\mathrm{sign}(h^a_i)=\xi^{\mu_a}_i$ there, we may write
\[
  |e^a_i|=\bigl|\tanh(\beta h^a_i)-\mathrm{sign}(h^a_i)\bigr|
        =1-\tanh\bigl(\beta|h^a_i|\bigr)
        =\frac{2}{e^{2\beta|h^a_i|}+1}
        \ \le\ 2e^{-2\beta|h^a_i|}\ \le\ 2e^{-2\beta\delta_\theta}=\varepsilon_\beta(\theta) ,
\]
where we used $1-\tanh x=2/(e^{2x}+1)\le 2e^{-2x}$ and then $|h^a_i|\ge\delta_\theta$. On the complementary event we only use the deterministic bound $|e^a_i|\le|t^a_i|+|\xi^{\mu_a}_i|\le2$. Splitting the expectation accordingly,
\[
  \E_{\mathcal S}|e^a_i|
  \le\varepsilon_\beta(\theta)\,\mathcal P(\mathcal E_i)+2\,\mathcal P(\mathcal E_i^{\,c})
  \le\varepsilon_\beta(\theta)+2\pi_\theta ,
\]
which is \eqref{eq:B_e_bound}. For the last assertion, the logistic rule \eqref{eq:Psigmagivenz} gives $\mathcal P\bigl((\sigma^a_i)''\neq\mathrm{sign}(h^a_i)\mid h^a_i\bigr)=(1+e^{2\beta|h^a_i|})^{-1}\le e^{-2\beta\delta_\theta}$ on $\mathcal E_i$; adding $\mathcal P(\mathcal E_i^{\,c})\le\pi_\theta$ concludes.
\end{proof}

\begin{lemma}\label{lem:overlap}
Let $Q^{\mu,b}:=\frac1N\SOMMA{j=1}{N}\bar\eta^\mu_j\,\xi^{\mu_b}_j$ and let $c_{\mu,\nu}:=\frac1N\SOMMA{j=1}{N}\xi^\mu_j\xi^\nu_j$ denote the empirical overlap of two archetypes. Then, at fixed $\bm\xi$,
\begin{equation}\label{eq:B_overlap}
  |Q^{\mu,b}|\le r^{-1} \ \text{ almost surely},
  \qquad
  \E_{\mathcal S}\bigl[Q^{\mu,b}\bigr]=c_{\mu,\mu_b},
  \qquad
  c_{\mu_b,\mu_b}=1 ,
\end{equation}
while, over the archetypes, for every $\mu\neq\nu$ and every $t>0$,
\begin{equation}\label{eq:B_overlap_tail}
  \mathcal P_{\bm\xi}\bigl(|c_{\mu,\nu}|\ge t\bigr)\le 2e^{-Nt^2/2},
  \qquad
  \mathcal P_{\bm\xi}\Bigl(\max_{\mu\le K,\,b\le L}\bigl|c_{\mu,\mu_b}-\delta_{\mu,\mu_b}\bigr|\ge t\Bigr)\le 2KL\,e^{-Nt^2/2} ,
\end{equation}
and consequently $\E_{\bm\xi}|c_{\mu,\nu}|\le N^{-1/2}$ for $\mu\neq\nu$.
\end{lemma}

\begin{proof}
The bound $|Q^{\mu,b}|\le r^{-1}$ follows from $|\bar\eta^\mu_j|\le r^{-1}$ and $|\xi^{\mu_b}_j|=1$. For the mean, $\E_{\mathcal S}[\bar\eta^\mu_j]=\xi^\mu_j$ by the definition of $\Delta^\mu_j$, whence $\E_{\mathcal S}[Q^{\mu,b}]=\frac 1N\sum_j\xi^\mu_j\xi^{\mu_b}_j=c_{\mu,\mu_b}$; in particular $c_{\mu_b,\mu_b}=1$ because $(\xi^{\mu_b}_j)^2=1$.

For \eqref{eq:B_overlap_tail}, fix $\mu\neq\nu$: the products $X_j=\xi^\mu_j\xi^\nu_j$ are independent and symmetric and take values in $[-1,1]$, so that $b_j-a_j=2$ and Lemma~\ref{lem:hoeffding} applied to $Nc_{\mu,\nu}=\sum_jX_j$ with $t\mapsto Nt$ gives $\mathcal P_{\bm\xi}(|c_{\mu,\nu}|\ge t)\le2\exp(-2N^2t^2/4N)=2e^{-Nt^2/2}$. The uniform bound follows from a union bound over the at most $KL$ pairs $(\mu,b)$ with $\mu\neq\mu_b$, the pairs with $\mu=\mu_b$ contribute $c_{\mu_b,\mu_b}-1=0$, and no independence among them is required. Finally, $\E_{\bm\xi}|c_{\mu,\nu}|\le(\E_{\bm\xi}c_{\mu,\nu}^2)^{1/2}=N^{-1/2}$ by Jensen's inequality.
\end{proof}

\begin{remark}\label{rem:quenched}
The $N^{-1/2}$ in Lemma~\ref{lem:overlap} is not an artefact of the algorithm: it is the residual overlap between two distinct archetypes at finite $N$, which the negative phase inevitably sees. It is the only contribution to Theorem~\ref{thm:fixed_point} that does not vanish exponentially, and it survives both the limits $\beta\to\infty$ and $M\to\infty$. Since it is a random function of the quenched disorder, it must be kept explicit in the bound and converted into a statement about $\bm\xi$ only at the very end, which is why Theorem~\ref{thm:fixed_point} is formulated with a confidence level $\gamma$.
\end{remark}

\subsection*{Proof of Theorem~\ref{thm:fixed_point} and of Corollary~\ref{cor:fixed_point}}

\begin{proof}[Proof of Theorem~\ref{thm:fixed_point}]
Fix $\theta\in(0,1)$ and set $\bm w=\bar{\bm\eta}$, so that $w^{\mu,b}_j=\bar\eta^\mu_j$ for every module $b$ and $T^b_\mu = \frac1N\sum_j\bar\eta^\mu_j t^b_j$; we nevertheless keep the module index visible, so that the modular sums can be followed. Taking the average $\E_{\mathcal S}$ of the exact identity \eqref{eq:B_reduction},
\begin{equation}\label{eq:B_pf_start}
  \E\bigl[\Delta y^{\mu,a}_i\bigr]
  =\underbrace{\E_{\mathcal S}\bigl[\bar s^{\mu_a}_i\bigr]\,\hat g^a_\mu}_{=:P}
  -\underbrace{\E_{\mathcal S}\Bigl[t^a_i\SOMMA{b=1}{L}g_{ab}T^b_\mu\Bigr]}_{=:T}
  -\underbrace{\frac1N\,\E_{\mathcal S}\Bigl[\bar\eta^\mu_i\bigl(1-(t^a_i)^2\bigr)\Bigr]}_{=:D} .
\end{equation}

\medskip
\noindent\textbf{Step 1: the positive phase.}
By Lemma~\ref{lem:majority} and \eqref{eq:B_sbar}, $\E_{\mathcal S}[\bar s^{\mu_a}_i]=r_B\,\xi^{\mu_a}_i$, so that
\begin{equation}\label{eq:B_pf_pos}
  P=r_B\,\xi^{\mu_a}_i\,\hat g^a_\mu=r_B\,\xi^{\mu_a}_i\SOMMA{b=1}{L}g_{ab}\,\delta_{\mu,\mu_b} .
\end{equation}
The majority vote is therefore what the positive phase contributes: an image of the archetype, faithful up to the factor $r_B$ of \eqref{eq:B_rB}. Had the visible layer been clamped to a single example, this factor would have been $r$, and no choice of the mini-batch size could have removed it.

\medskip
\noindent\textbf{Step 2: the diagonal correction $D$ is negligible.}
Writing $t^a_i=\xi^{\mu_a}_i+e^a_i$ and using $(\xi^{\mu_a}_i)^2=1$,
\[
  1-(t^a_i)^2=-2\,\xi^{\mu_a}_ie^a_i-(e^a_i)^2,
  \qquad\text{so}\qquad
  \bigl|1-(t^a_i)^2\bigr|\le2|e^a_i|+|e^a_i|^2\le4|e^a_i|,
\]
the last step using $|e^a_i|\le2$. Since $|\bar\eta^\mu_i|\le r^{-1}$, Lemma~\ref{lem:recon} gives
\begin{equation}\label{eq:B_pf_D}
  |D|\ \le\ \frac{4}{rN}\,\E_{\mathcal S}|e^a_i|\ \le\ \frac{4\bigl(\varepsilon_\beta(\theta)+2\pi_\theta\bigr)}{rN}.
\end{equation}

\medskip
\noindent\textbf{Step 3: expanding $T$, and the cancellation.}
Substituting $t^a_i=\xi^{\mu_a}_i+e^a_i$ and $t^b_j=\xi^{\mu_b}_j+e^b_j$ in $T$ and expanding the product gives four contributions,
\[
  t^a_it^b_j
  =\underbrace{\xi^{\mu_a}_i\xi^{\mu_b}_j}_{\rm(I)}
  +\underbrace{\xi^{\mu_a}_i\,e^b_j}_{\rm(II)}
  +\underbrace{e^a_i\,\xi^{\mu_b}_j}_{\rm(III)}
  +\underbrace{e^a_i\,e^b_j}_{\rm(IV)} ,
\]
which we treat separately.

\emph{(I) --- the leading term.} Recognising the overlap of Lemma~\ref{lem:overlap}, $Q^{\mu,b}=\frac1N\sum_j\bar\eta^\mu_j\xi^{\mu_b}_j$,
\[
  \frac1N\SOMMA{b}{}g_{ab}\SOMMA{j}{}\E_{\mathcal S}\bigl[\bar\eta^\mu_j\xi^{\mu_a}_i\xi^{\mu_b}_j\bigr]
  =\xi^{\mu_a}_i\SOMMA{b}{}g_{ab}\,\E_{\mathcal S}\bigl[Q^{\mu,b}\bigr]
  =\xi^{\mu_a}_i\SOMMA{b=1}{L}g_{ab}\,c_{\mu,\mu_b} ,
\]
so that, subtracting \eqref{eq:B_pf_pos}, adding and removing $\xi^{\mu_a}_i\hat g^a_\mu$ and using $c_{\mu_b,\mu_b}=1$ together with $|\hat g^a_\mu|\le1$,
\begin{equation}\label{eq:B_pf_I}
  \bigl|P-\mathrm{(I)}\bigr|
  =\Bigl|\xi^{\mu_a}_i\SOMMA{b}{}g_{ab}\bigl(\delta_{\mu,\mu_b}-c_{\mu,\mu_b}\bigr)-(1-r_B)\,\xi^{\mu_a}_i\,\hat g^a_\mu\Bigr|
  \le (1+\Lambda)\max_{b\le L}\bigl|c_{\mu,\mu_b}-\delta_{\mu,\mu_b}\bigr|+(1-r_B) .
\end{equation}
This is the cancellation the theorem rests on, and it is worth reading it once more in words: the clamp asks the hidden layer for the value $\hat z^a_\mu$, and the machine, having reconstructed the archetypes on its visible layers, returns exactly $\hat z^a_\mu$, the weights being centred on the archetypes and not on a shrunken copy of them. The two agree once multiplied by the visible activity, which the majority vote supplies with the amplitude $r_B$: the residual mismatch $1-r_B$ is the price of the fact that a cue is a binary vector, and it is the mini-batch channel of \eqref{eq:thm_bound}. This is also where the two ingredients of the protocol are seen to be inseparable. With biased weights, both sides would have carried a factor $r$ and would have cancelled; with unbiased weights, the positive phase must be cleaned as well, and the majority vote is what cleans it.

\emph{(II).} Using $|\xi^{\mu_a}_i|=1$, $|\bar\eta^\mu_j|\le r^{-1}$ and Lemma~\ref{lem:recon} at each site,
\[
  \Bigl|\frac1N\SOMMA{b}{}g_{ab}\SOMMA{j}{}\E_{\mathcal S}\bigl[\bar\eta^\mu_j\xi^{\mu_a}_ie^b_j\bigr]\Bigr|
  \le\frac{1}{rN}\SOMMA{b}{}|g_{ab}|\SOMMA{j}{}\E_{\mathcal S}\bigl|e^b_j\bigr|
  \le\frac{1+\Lambda}{r}\bigl(\varepsilon_\beta(\theta)+2\pi_\theta\bigr),
\]
since $\E_{\mathcal S}|e^b_j|\le\varepsilon_\beta(\theta)+2\pi_\theta$ uniformly in $j$ and the $N$ terms are divided by $N$.

\emph{(III).} Here $e^a_i$ does not depend on the summation index $j$, so the sum reconstructs the overlap again:
\[
  \frac1N\SOMMA{b}{}g_{ab}\SOMMA{j}{}\E_{\mathcal S}\bigl[\bar\eta^\mu_je^a_i\xi^{\mu_b}_j\bigr]
  =\E_{\mathcal S}\Bigl[e^a_i\SOMMA{b}{}g_{ab}Q^{\mu,b}\Bigr],
  \qquad
  \Bigl|\SOMMA{b}{}g_{ab}Q^{\mu,b}\Bigr|\le \frac{1+\Lambda}{r}
\]
by the deterministic bound $|Q^{\mu,b}|\le r^{-1}$ of Lemma~\ref{lem:overlap}. The contribution is therefore at most $\frac{1+\Lambda}{r}\,\E_{\mathcal S}|e^a_i|\le\frac{1+\Lambda}{r}\bigl(\varepsilon_\beta(\theta)+2\pi_\theta\bigr)$.

\emph{(IV).} Bounding one factor by $|e^a_i|\le2$ and averaging the other,
\[
  \Bigl|\frac1N\SOMMA{b}{}g_{ab}\SOMMA{j}{}\E_{\mathcal S}\bigl[\bar\eta^\mu_je^a_ie^b_j\bigr]\Bigr|
  \le\frac{2}{rN}\SOMMA{b}{}|g_{ab}|\SOMMA{j}{}\E_{\mathcal S}\bigl|e^b_j\bigr|
  \le\frac{2(1+\Lambda)}{r}\bigl(\varepsilon_\beta(\theta)+2\pi_\theta\bigr).
\]

\medskip
\noindent\textbf{Step 4: collecting the estimates.}
Inserting the four contributions and \eqref{eq:B_pf_D} into \eqref{eq:B_pf_start}, the leading terms cancel by (I) and we are left with
\begin{align*}
  \bigl|\E[\Delta y^{\mu,a}_i]\bigr|
  &\le\underbrace{(1+\Lambda)\max_{b}\bigl|c_{\mu,\mu_b}-\delta_{\mu,\mu_b}\bigr|+(1-r_B)}_{\rm(I)}
   +\underbrace{\frac{1+\Lambda}{r}\bigl(\varepsilon_\beta(\theta)+2\pi_\theta\bigr)}_{\rm(II)}
   +\underbrace{\frac{1+\Lambda}{r}\bigl(\varepsilon_\beta(\theta)+2\pi_\theta\bigr)}_{\rm(III)}\\
   &\quad+\underbrace{\frac{2(1+\Lambda)}{r}\bigl(\varepsilon_\beta(\theta)+2\pi_\theta\bigr)}_{\rm(IV)}
   +\underbrace{\frac{4\bigl(\varepsilon_\beta(\theta)+2\pi_\theta\bigr)}{rN}}_{D},
\end{align*}
that is, for $N\ge4$, so that $4/N\le1\le1+\Lambda$,
\begin{equation}\label{eq:B_pf_final}
  \bigl|\E[\Delta y^{\mu,a}_i]\bigr|
  \ \le\ \frac{5(1+\Lambda)}{r}\bigl(\varepsilon_\beta(\theta)+2\pi_\theta\bigr)
        +(1-r_B)
        +(1+\Lambda)\max_{b\le L}\bigl|c_{\mu,\mu_b}-\delta_{\mu,\mu_b}\bigr| .
\end{equation}
Substituting $\varepsilon_\beta(\theta)=2e^{-2(1-\theta)\beta (1-\Lambda)}$ from Lemma~\ref{lem:recon}, $\pi_\theta\le2Le^{-c_1(\theta)/\rho}$ from Lemma~\ref{lem:margin} and $1-r_B\le2e^{-Br^2/2}$ from Lemma~\ref{lem:majority} yields
\begin{equation}\label{eq:B_pf_theta}
  \bigl|\E[\Delta y^{\mu,a}_i]\bigr|
  \ \le\ \frac{20(1+\Lambda)}{r}\Bigl(e^{-2(1-\theta)\beta (1-\Lambda)}+L\,e^{-c_1(\theta)/\rho}\Bigr)
  +2e^{-Br^2/2}
  +(1+\Lambda)\max_{b\le L}\bigl|c_{\mu,\mu_b}-\delta_{\mu,\mu_b}\bigr| ,
\end{equation}
where we used $10\le20$ to write the two exponential channels with a single constant. Since $\theta\in(0,1)$ was arbitrary, \eqref{eq:B_pf_theta} holds for every $\theta$; the symmetric choice $\theta=1/2$ gives $2(1-\theta)=1$ and $c_1(1/2)=c_1$, that is precisely the two exponents appearing in \eqref{eq:thm_bound}, with $C=20(1+\Lambda)/r$ and $C'=1+\Lambda$.  The whole argument has been carried out at fixed $\mu_1,\dots,\mu_L$, and none of the three constants depends on that choice, the finite-size term being already controlled uniformly over the $KL$ pairs by \eqref{eq:B_overlap_tail}. The bound is therefore uniform in $\hat{\bm\mu}$ and survives the average over it contained in $\mathcal Q$.

It remains to convert the last term of \eqref{eq:B_pf_theta} into a statement about the archetypes. By the uniform tail \eqref{eq:B_overlap_tail}, for every $\gamma>0$ the event
\[
  \Bigl\{\max_{\mu\le K,\,b\le L}\bigl|c_{\mu,\mu_b}-\delta_{\mu,\mu_b}\bigr|<\gamma\Bigr\}
\]
has $\mathcal P_{\bm\xi}$-probability at least $1-2KL\,e^{-N\gamma^2/2}$, and on that event the right-hand side of \eqref{eq:B_pf_theta} may be maximised over $i,\mu,a$. Recalling $\Delta w^{\mu,a}_i=\epsilon\beta\,\Delta y^{\mu,a}_i$ gives \eqref{eq:thm_bound}.
\end{proof}

\begin{remark}\label{rem:fluctuations}
Theorem~\ref{thm:fixed_point} is a statement about the mean of a single update, not about its fluctuations, which remain of order one even where the mean vanishes: the positive phase is built on one mini-batch and the negative one on a single Gibbs sweep, and neither is averaged over. At a finite learning rate the weights therefore perform a stochastic approximation around $\bar{\bm\eta}$ rather than sitting exactly on it, with a stationary spread that shrinks with $\epsilon\beta$. This is the usual situation of stochastic gradient methods, and it is why the fixed point is to be understood in the sense of Definition~\ref{def:fixed_point}.
\end{remark}

\begin{proof}[Proof of Corollary~\ref{cor:fixed_point}]
The choice $\gamma=\sqrt{2\ln(2KLN)/N}$ makes $2KL\,e^{-N\gamma^2/2}=N^{-1}$, which gives the stated bound and confidence level. The first term of the bound vanishes as $\beta\to\infty$, the second as $M\to\infty$ (i.e.\ $\rho\to0$), the third as $B\to\infty$, and the fourth as $N\to\infty$ provided $\ln K=o(N)$; the four channels being independent, the order of the limits is immaterial. Finally, taking $\E_{\bm\xi}$ of \eqref{eq:B_pf_theta} at fixed $i,\mu,a$ and using $\E_{\bm\xi}|c_{\mu,\nu}|\le N^{-1/2}$ from Lemma~\ref{lem:overlap}, together with the fact that the diagonal pairs contribute zero, replaces the last term by $C'/\sqrt N$.
\end{proof}


\begin{thebibliography}{20}
\providecommand{\natexlab}[1]{#1}
\providecommand{\url}[1]{\texttt{#1}}
\expandafter\ifx\csname urlstyle\endcsname\relax
  \providecommand{\doi}[1]{doi: #1}\else
  \providecommand{\doi}{doi: \begingroup \urlstyle{rm}\Url}\fi

\bibitem[Zdeborov(2017)]{Lenka-2017}
Lenka Zdeborov\'a.
\newblock Machine learning: New tool in the box.
\newblock \emph{Nature Physics}, 13\penalty0 (5):\penalty0 420--421, 2017.
\newblock \doi{10.1038/nphys4053}.

\bibitem[Hopfield(1982)]{hopfield1982}
John~J. Hopfield.
\newblock Neural networks and physical systems with emergent collective computational abilities.
\newblock \emph{Proceedings of the National Academy of Sciences}, 79\penalty0 (8):\penalty0 2554--2558, 1982.
\newblock \doi{10.1073/pnas.79.8.2554}.

\bibitem[Amit et~al.(1985)Amit, Gutfreund, and Sompolinsky]{amitgutfreundsompolinsky1985}
Daniel~J. Amit, Hanoch Gutfreund, and Haim Sompolinsky.
\newblock Spin-glass models of neural networks.
\newblock \emph{Physical Review A}, 32\penalty0 (2):\penalty0 1007--1018, 1985.
\newblock \doi{10.1103/PhysRevA.32.1007}.

\bibitem[Ackley et~al.(1985)Ackley, Hinton, and Sejnowski]{ackleyhintonsejnowski1985}
David~H. Ackley, Geoffrey~E. Hinton, and Terrence~J. Sejnowski.
\newblock A learning algorithm for {Boltzmann} machines.
\newblock \emph{Cognitive Science}, 9\penalty0 (1):\penalty0 147--169, 1985.
\newblock \doi{10.1207/s15516709cog0901_7}.

\bibitem[Hinton(2002)]{hinton2002cd}
Geoffrey~E. Hinton.
\newblock Training products of experts by minimizing contrastive divergence.
\newblock \emph{Neural Computation}, 14\penalty0 (8):\penalty0 1771--1800, 2002.
\newblock \doi{10.1162/089976602760128018}.

\bibitem[Coolen et~al.(2005)Coolen, K{\"u}hn, and Sollich]{10.1093/oso/9780198530237.001.0001}
Anthony C.~C. Coolen, Reimer K{\"u}hn, and Peter Sollich.
\newblock \emph{Theory of Neural Information Processing Systems}.
\newblock Oxford University Press, 2005.
\newblock \doi{10.1093/oso/9780198530237.001.0001}.

\bibitem[Barra et~al.(2012)Barra, Bernacchia, Santucci, and Contucci]{BARRA20121}
Adriano Barra, Alberto Bernacchia, Enrica Santucci, and Pierluigi Contucci.
\newblock On the equivalence of {Hopfield} networks and {Boltzmann} machines.
\newblock \emph{Neural Networks}, 34:\penalty0 1--9, 2012.
\newblock \doi{10.1016/j.neunet.2012.06.003}.

\bibitem[Agliari et~al.(2012)Agliari, Barra, Galluzzi, Guerra, and Moauro]{Agliari_2012PRL}
Elena Agliari, Adriano Barra, Andrea Galluzzi, Francesco Guerra, and Francesco Moauro.
\newblock Multitasking associative networks.
\newblock \emph{Physical Review Letters}, 109\penalty0 (26):\penalty0 268101, 2012.
\newblock \doi{10.1103/PhysRevLett.109.268101}.

\bibitem[M{\'e}zard(2017)]{Mezard-PRE2017}
Marc M{\'e}zard.
\newblock Mean-field message-passing equations in the {Hopfield} model and its generalizations.
\newblock \emph{Physical Review E}, 95\penalty0 (2):\penalty0 022117, 2017.
\newblock \doi{10.1103/PhysRevE.95.022117}.

\bibitem[Tubiana and Monasson(2017)]{TubianaMonasson-PRL2017}
Jer\^ome Tubiana and R\'emi Monasson.
\newblock Emergence of compositional representations in restricted {Boltzmann} machines.
\newblock \emph{Physical Review Letters}, 118\penalty0 (13):\penalty0 138301, 2017.
\newblock \doi{10.1103/PhysRevLett.118.138301}.

\bibitem[Decelle and Furtlehner(2021)]{DecelleFurtlehner-2021}
Aur{\'e}lien Decelle and Cyril Furtlehner.
\newblock Restricted {Boltzmann} machine: Recent advances and mean-field theory.
\newblock \emph{Chinese Physics B}, 30\penalty0 (4):\penalty0 040202, 2021.
\newblock \doi{10.1088/1674-1056/abd160}.

\bibitem[di~Sarra et~al.(2025)di~Sarra, Bravi, and Roudi]{Roudi-EPL2025}
Giovanni di~Sarra, Barbara Bravi, and Yasser Roudi.
\newblock The unbearable lightness of restricted {Boltzmann} machines: Theoretical insights and biological applications.
\newblock \emph{Europhysics Letters}, 149\penalty0 (2):\penalty0 21001, 2025.
\newblock \doi{10.1209/0295-5075/ada636}.

\bibitem[Alemanno et~al.(2023)Alemanno, Aquaro, Kanter, Barra, and Agliari]{Alemanno_2023}
Francesco Alemanno, Miriam Aquaro, Ido Kanter, Adriano Barra, and Elena Agliari.
\newblock Supervised {Hebbian} learning.
\newblock \emph{Europhysics Letters}, 141\penalty0 (1):\penalty0 11001, 2023.
\newblock \doi{10.1209/0295-5075/aca55f}.

\bibitem[Leonelli et~al.(2021)Leonelli, Agliari, Albanese, and Barra]{LEONELLI2021314}
Francesca~E. Leonelli, Elena Agliari, Linda Albanese, and Adriano Barra.
\newblock On the effective initialisation for restricted {Boltzmann} machines via duality with {Hopfield} model.
\newblock \emph{Neural Networks}, 143:\penalty0 314--326, 2021.
\newblock \doi{10.1016/j.neunet.2021.06.017}.

\bibitem[Smart and Zilman(2021)]{SmartZilman-ICLR2021}
Matthew Smart and Anton Zilman.
\newblock On the mapping between {Hopfield} networks and restricted {Boltzmann} machines.
\newblock In \emph{International Conference on Learning Representations (ICLR)}, 2021.
\newblock arXiv:2101.11744.

\bibitem[Agliari et~al.(2026{\natexlab{a}})Agliari, Alessandrelli, Barra, Centonze, and Ricci-Tersenghi]{AGLIARI2026108181}
Elena Agliari, Andrea Alessandrelli, Adriano Barra, Martino~Salomone Centonze, and Federico Ricci-Tersenghi.
\newblock Networks of {Hebbian} networks: more is different.
\newblock \emph{Neural Networks}, 194:\penalty0 108181, 2026{\natexlab{a}}.
\newblock \doi{10.1016/j.neunet.2025.108181}.

\bibitem[Fachechi et~al.(2025)Fachechi, Agliari, Alessandrelli, and Mourao]{agliari2025multi}
Alberto Fachechi, Elena Agliari, Andrea Alessandrelli, and Paulo~Duarte Mour{\~a}o.
\newblock Multi-channel pattern reconstruction through {$L$}-directional associative memories.
\newblock In \emph{New Frontiers in Associative Memories}, 2025.

\bibitem[Agliari et~al.(2025)Agliari, Alessandrelli, Barra, Centonze, and Ricci-Tersenghi]{agliari2025generalized}
Elena Agliari, Andrea Alessandrelli, Adriano Barra, Martino~Salomone Centonze, and Federico Ricci-Tersenghi.
\newblock Generalized hetero-associative neural networks.
\newblock \emph{Journal of Statistical Mechanics: Theory and Experiment}, 2025\penalty0 (1):\penalty0 013302, 2025.
\newblock \doi{10.1088/1742-5468/ada094}.

\bibitem[Agliari et~al.(2026{\natexlab{b}})Agliari, Fachechi, and Mour{\~a}o]{AGLIARI2026131134}
Elena Agliari, Alberto Fachechi, and Paulo~Duarte Mour{\~a}o.
\newblock The beneficial role of noises for disentanglement tasks in modular {Hebbian} networks.
\newblock \emph{Physica A: Statistical Mechanics and its Applications}, 682:\penalty0 131134, 2026{\natexlab{b}}.

\bibitem[Ernoult et~al.(2019)Ernoult, Grollier, and Querlioz]{Ernoult-SciRep2019}
Maxence Ernoult, Julie Grollier, and Damien Querlioz.
\newblock Using memristors for robust local learning of hardware restricted {Boltzmann} machines.
\newblock \emph{Scientific Reports}, 9:\penalty0 1851, 2019.
\newblock \doi{10.1038/s41598-018-38181-3}.

\end{thebibliography}
\end{document}